\documentclass[11pt]{article}
\usepackage{amsfonts,amsthm,amsmath,amssymb,latexsym,comment}
\ifx\pdftexversion\undefined
  \usepackage[a4paper,colorlinks,linkcolor=black,filecolor=black,citecolor=black,urlcolor=black,pdfstartview=FitH]{hyperref}
\else
  \usepackage[a4paper,colorlinks,linkcolor=blue,filecolor=blue,citecolor=blue,urlcolor=blue,pdfstartview=FitH]{hyperref}
\fi
\usepackage{chicagor}
\usepackage{epsfig}
\usepackage{graphicx}
\usepackage{hhline}
\usepackage{tabularx}
\usepackage[ruled, linesnumbered]{algorithm2e}

\SetCommentSty{mycommfont}
\newtheorem{theorem}{Theorem}

\newtheorem{lemma}{Lemma}
\newtheorem{definition}{Definition}
\newtheorem{proposition}{Proposition}
\newcommand{\commentout}[1]{}

\newcommand{\longvec}{\overrightarrow}

\newcommand{\VSS}{\mathrm{VSS}}

\newcommand{\medProtocol}{\vec{\tau}^f + \tau^f_d}

\def\Sch{\mathit{sch}}
\def\Snd{\mathit{snd}}
\def\Rec{\mathit{rec}}
\def\Done{\mathit{done}}
\def\Comp{\mathit{comp}}
\def\Out{\mathit{out}}
\def\BCGsch{xxyyzz }
\newcommand{\proceed}{\mathit{proceed}}
\renewcommand{\next}{\mathit{next}}

\begin{document}


\title{Security in Asynchronous Interactive Systems}
\author{Ivan Geffner\thanks{Supported in part by NSF grant IIS-1703846.}\\ ieg8@cornell.edu\\ Cornell University\and 
Joseph Y. Halpern\thanks{Supported in part by NSF grants IIS-1703846 and IIS-1718108,
ARO grant W911NF-17-1-0592, and a grant 
from Open Philanthropy.}\\  halpern@cs.cornell.edu \\ Cornell University}
\maketitle

\begin{abstract}
Secure function computation has been thoroughly studied and optimized
in the past decades. We extend techniques used for
secure computation to simulate arbitrary protocols involving a
mediator. The key feature of our notion of simulation is that it is 
bidirectional: not only does the simulation produce only outputs that
could happen in the original protocol, but the simulation produces
all such outputs. In a synchronous system, it can be shown
that this requirement can already be achieved by the standard notion
of secure computation. However, in an
asynchronous system, new subtleties arise because the scheduler can
influence the output.  We provide a construction that
is secure if  $n > 4t$, where $t$ is the number malicious agents,
which is provably the best possible.  We also 
show that our construction satisfies 
additional security properties even if $3t <  n \le 4t$.
\end{abstract}

\section{Introduction}\label{sec:intro}
In a distributed system, agents often want to be able to carry out
a computation without revealing any private information.  There has
been a great deal of work showing how and to what extent this can be
done.  We briefly review the most relevant work here.

Ben-Or, Goldwasser and Widgerson 
\citeyear{bgw} 
(BGW from
now on) 
showed that, if $n > 3t$, then every function $f$ of $n$ inputs can be
securely computed by $n$ 
agents in a synchronous system with private communication channels, 
where ``securely computed'' means that no coalition of at 
most $t$ malicious agents can either
(a) prevent the honest agents from correctly computing the output of $f$ given
their inputs (assuming some fixed inputs for malicious agents who do
not provide inputs) or (b) learn anything
about the inputs of the honest agents (beyond what can be concluded
from the output of $f$).
The notion of an agent  ``not learning anything'' is formalized by
comparing what happens in the actual computation to
what could have happened had there been a trusted third party (which
we here call a \emph{mediator}) who will calculate $f(x_1, \ldots,
x_n)$ after being given the input $x_i$ by agent $i$, for $i = 1,
\ldots, n$.  Then, roughly speaking, the malicious agents do not learn
anything if the distribution of outputs in the actual computation
could have also resulted in the computation with a mediator if the
malicious agents had given the appropriate input to the mediator.  


Ben-Or, Canetti and Goldreich 
\citeyear{BCG93}
 (BCG from now on) proved
analogous results 
in the asynchronous case.
Asynchrony raises new subtleties.  For example, agent $i$ cannot
tell if the fact that he has received no messages
from another agent $j$ (which means that $i$ cannot use $j$'s
input in computing $f$) is due to the fact that $j$ is malicious or
that its messages have not yet arrived.
Roughly speaking, when defining secure function computation in an
asynchronous setting,
BCG require that for every scheduler $\sigma_e$ and set $T$ of
malicious agents, no matter what the agents in $T$ do, the resulting
distribution over outputs could have also resulted in the computation
with a mediator if the malicious agents had given the appropriate
input to the mediator.   

BCG show that, in asynchronous systems, if $n > 4t$, the
malicious agents cannot prevent the honest agents from correctly
computing the output 
of $f$ given their inputs, nor can the malicious agents learn anything
about the inputs of the honest agents.
Ben-Or, Kelmer and Rabin 
\citeyear {BKR94}
 (BKR from now on) then showed 
 if we are willing to tolerate a small probability $\epsilon
> 0$
that the agents do not correctly compute $f$ or that the malicious
agents learn something, then  we can achieve this if $n > 3t$.  
BCG and BKR also prove matching lower bounds for their  results,
showing that we really need to have $n > 4t$ (resp., $n  > 3t$).

We can view secure function computation as a one-round interaction
with a trusted mediator: each agent sends its input to the mediator,
the mediator waits until it receives enough inputs, applies $f$ to
these inputs (again, replacing missing inputs with a default
value), and sends the output back to the agents, who then output it.
We generalize BCG and BKR's results for
function computation to a more general setting. Specifically, we want to 
simulate arbitrary interactions with a mediator, not just function
computation.   
Also, unlike previous approaches, we want the simulation to be
``bidirectional'': the set of possible output distributions that arise
with the mediator must be the same as those that arise without 
the mediator, even in the presence of malicious
parties. 
More precisely,
we show that, given a protocol
$\vec{\pi}$ for $n$ agents and a protocol $\pi_d$ for a mediator,
we can construct a protocol $\vec{\pi}'$ such that
for all sets $T$ of fewer than $n/4$ malicious agents,
the following properties hold:
\begin{itemize}
\item [(a)]
   For all protocols $\vec{\tau}_T'$ for the 
malicious agents and all schedulers
$\sigma_e'$ in the setting without the mediator, there
exists a protocol $\vec{\tau}_T$ for the agents in $T$ and a
scheduler $\sigma_e$ in the setting with the mediator 
such that, for all input profiles $\vec{x}$, the output distribution in the
computation with 
$\vec{\pi}'$, $\vec{\tau}'$, and $\sigma_e'$ with input $\vec{x}$ is
the same as the output 
distribution with $\vec{\pi} + \pi_d$, $\vec{\tau}$, and $\sigma_e$
with input $\vec{x}$.
\item [(b)]
  For all protocols $\vec{\tau}_T$ for the 
  malicious agents and all schedulers
$\sigma_e$ in the setting with the mediator, there
  exists a protocol $\vec{\tau}'_T$ for the agents in $T$ and a
scheduler $\sigma_e'$ in the setting without a mediator
such that, for all input profiles $\vec{x}$, the output distribution in the
computation with 
$\vec{\pi}'$, $\vec{\tau}'$, and $\sigma_e'$ with input $\vec{x}$ is
the same as the output 
distribution with $\vec{\pi} + \pi_d$, $\vec{\tau}$, and $\sigma_e$
with input $\vec{x}$.
\end{itemize}
This result implies that arbitrary distributed
protocols that work in the presence of a trusted mediator can be
compiled to protocols that work without a mediator, as 
long as there are less than $n/4$ malicious agents.
And, just as BKR, if we allow a probability $\epsilon$ of error, we
can get this result while tolerating up to $n/3$ malicious agents.
BCG proved the analogue of 
(a) for secure function computation,
 which is enough for security purposes:
if there is any bad behavior in the protocol without the mediator,
this bad behavior must already exist in the protocol with the mediator.
However, 
(b)
 also seems like a natural requirement; if a protocol
satisfies this property, then all behaviors in the protocol with the
mediator also  occur in the protocol without the mediator.

Clearly, the results of BCG and BKR are special cases of our result.
However, 
in general,
our results
do not follow from those of BCG/BKR,
as is shown in Section~\ref{sec:secure-computation-and-mediators}.
Specifically, the results of BCG/BKR do not give us property 
(b),
since the outcome can depend on the behavior of the scheduler.
For example, consider a protocol for two agents and a mediator $m$ in
which each 
agent sends its input to the mediator, the mediator $m$ sends to
each agent the first message it receives, and each agent outputs
whatever they receive from the mediator. 
Let $\sigma_e^i$ be the scheduler that delivers the message from
agent $i$ first, for $i=1,2$.
It is easy to check
that if the agents have inputs $0$ and $1$, respectively, and play with
mediator $\sigma_e^1$, then they both output $0$, while if they play
with $\sigma_e^2$, then they both output $1$. This means that, unlike
secure function computation, even if all the
agents are honest,
the distribution over the
agents' outputs can depend on the scheduler's protocol, not just the
agents' inputs.  
Even though
our results do not follow from those of BCG/BKR, our
proofs very much follow the lines of those of BCG/BKR. However, there
are some new subtleties that arise in our setting. 
\commentout{
In particular, to ensure that in their simulation all possible
schedulers in the mediator game are being taken into consideration,
agents must use the information received from the actual scheduler
(for example, when are they being scheduled and the order in which
messages are being delivered) as feedback for how the scheduler would
act in the mediator game. Thus, besides their own inputs, agents
might use additional data about the current scheduler's behavior in
their computations. 
}
In particular, as the example above shows, when we try to implement
the setting with the mediator, the agents must somehow keep track
of the scheduler's possible behaviors.  Doing this adds nontrivial
complexity to our argument.

Besides the main result, we also show that our protocol without the
mediator has two additional security properties, which may be of
independent interest. 
 Specifically, we show that the
 following two properties hold for coalitions of malicious agents of size at
 most $t < n/3$.  
\begin{itemize}
\item [(P1)]
 The only way malicious agents can disrupt the computation
  is by preventing honest agents from terminating; if an
  honest agent terminates, then its output is correct. 
\item [(P2)]
 If $2t+1$ or more honest agents terminate, then all honest
  agents terminate.
  That is, either all the honest agents terminate or a nontrivial
  number of honest agents (more than $n-2t$) do not terminate.
\end{itemize}
If we allow an $\epsilon$ probability of error, we get analogous
results if we have $n > 2t$ rather than $n > 3t$.
We remark that these two properties 
 are in fact also satisfied by BCG's and BKR's
 implementations, but they do not prove this (or even state the
 properties explicitly).


 Our interest in these properties stems in part from 
a  game-theoretic variant of the problem that we consider 
a companion paper 
\cite{ADGH19}
 where agents get utility for
various outcomes, and, in addition to honest and malicious agents,
there are \emph{rational agents}, who will deviate from a protocol if
(and only if) it is to their benefit to do so.  We also assume that 
honest agents can leave ``wills'', so that if sufficiently many
honest agents do not terminate, the remaining agents will be
punished.  The second property above guarantees that either all the honest
agents terminate, or sufficiently many of them do not terminate to
guarantee that rational agents will not try to prevent honest agents
from terminating (due to the threat of punishment).   The first property
above guarantees that if all   the honest agents terminate, their
output will be correct.  Thus, using these results allows us to obtain
results stronger than those of this paper in the game-theoretic setting. 

The focus of this paper is on upper bounds.  Since our algorithms have
the same upper bounds as those of BCG and BKR, despite the results of
BCG and BKR being special cases of our results, and BCG and BKR prove 
lower bounds that match their upper bonds on the number of malicious
agents that can be tolerated, we immediately get lower bounds that
match our upper bounds from the results   of BCG and BKR.

\section{The Model}\label{sec:model}

The model used throughout this paper is that of an asynchronous
network in which every pair of agents can communicate through a
private and reliable communication channel.
For most of our results, we assume that all messages sent through any
of these channels are eventually received, but they can be delayed
arbitrarily. The order in which these messages are 
received is determined by the 
\emph{environment} (also called the \emph{scheduler}). The scheduler also
chooses the order in which the agents are scheduled.
For some of the results of this paper, we drop the condition that all
messages must be eventually delivered.  We call these more general schedulers
\emph{relaxed schedulers}.

Whenever a
agent is scheduled, it reads all the messages that it has received since
the last time it was scheduled,
sends a (possibly empty) sequence of messages, and then performs some
internal actions. 
We assume that the scheduler does not deliver any
message or schedule other agents 
during an agent's turn.
Thus, although agent $i$ does not send all its messages simultaneously
when it is scheduled, they are sent atomically, in the sense that
no other agent is scheduled while $i$ is scheduled, nor are any
messages delivered while $i$ is scheduled.
Note that the atomicity assumption is really a constraint on the
scheduler's protocol.

More precisely,
consider the following types of \emph{events}: 
\begin{itemize}
\item $\Sch(i)$: Agent $i$ gets scheduled.
  \item $\Snd(\mu, j,i)$: Agent $i$ sends a message $\mu$ to agent $j$.

\item $\Rec(\mu,j,i)$: Message $\mu$ sent by $j$ is received by $i$. The
message $\mu$ must be one sent at an earlier time to $i$ that 
was not already received.
\item $\Comp(v,i)$: Agent $i$ locally computes value $v$.
  \item $\Out(s,i)$: Agent $i$ outputs string $s$.
  \item $\Done(i)$: $i$ is done sending messages and performing
    computations (for now).
\end{itemize}
For simplicity, we assume that agents can output only strings in
$\{0,1\}^*$. Note that all countable sets can be encoded by such
strings, and thus we can freely talk about players being able to
output any element of any countable set (for instance, elements of a
finite field $\mathbb{F}_q$) by assuming that they are actually
outputting an encoding of these elements.
We
also 
 assume that at most one event occurs at each time step.
Let 
$h(m)$ 
denote a \emph{global history} up to time $m$: a sequence that starts
with an input profile $\vec{x}$, followed by
 the ordered sequence of events that have occurred
up to and including time $m$.  We assume that the only events between
events of the form $\Sch(i)$ and $\Done(i)$ are ones of the form
$\Snd(\mu,j,i)$ and $\Comp(v,i)$.  This captures our atomicity
assumption.  We do not 
include explicit events that correspond to reading messages. 
(Nothing would change if we included
them; they would simply clutter the notation.)  Message delivery
(which is assumed to be under the control of the scheduler)
occurs at times between when agents are scheduled.
We can also consider the subsequence
involving 
agent $i$, namely, $i$'s initial state, followed by events of the form
$\Sch(i)$, $\Snd(\cdot, \cdot, 
i)$, $\Comp(\cdot, i)$, $\Rec(\cdot, \cdot, i)$, and $\Done(i)$. This
subsequence is called $i$'s \emph{local history}. 
We drop the argument $m$ if can be deduced from context or if it is
not relevant (for instance, when we consider the local history of
an agent after a particular event). 

Agent $i$ moves only after a $\Sch(i)$ event.  What it does (in
particular, the order in which $i$ sends messages) is determined by
$i$'s protocol, which is a function of $i$'s local history.
The scheduler moves after an action of the form $\Done(i)$ or
$\Rec(\cdot,\cdot,i)$.  It is convenient to assume that the
scheduler is also running a protocol, which is also a function of its local
history.   Since the scheduler does not see the contents of messages,
we can take its history to be identical to 
$h(m)$,
 except that the messages are
removed, although we do track the index of the messages delivered;
that is, we replace events of the form $\Snd(\mu,i,j)$ and 
$\Rec(\mu,i,j)$ by 
$\Snd(i,j)$ and $\Rec(i,j,\ell)$, where $\ell$ is the index of the
message sent by $i$ to $j$ in 
$h(m)$.
 For instance, $\Rec(i,j,2)$
means that the second message sent by $i$ to $j$ was delivered to
$j$.
Note that the scheduler does see events of the form $\Done(i)$;
indeed, these are signals to the scheduler that it can move, since
$i$'s turn is over.
%
Since we view the agents (and the mediator) as sending messages atomically,
\commentout{
we can split a  history into \emph{blocks}, each consisting of a
subsequence of consecutive events. We consider two kinds of blocks: 
\begin{itemize}
\item \emph{Scheduler blocks}, which consist of a sequence of $\Rec$
  events followed by a single $\Sch$ event. 
  \item \emph{Agent $i$ blocks}, which consist of a sequence of
events of the form $\Snd(\cdot,\cdot,i)$, followed by a single
$\Done(i)$ event.   
\end{itemize} 
We can view an agent $i$ block as consisting of all the actions
performed by $i$ when it is scheduled once; similarly, 
a scheduler block consists of all the actions performed by the scheduler
between when agents are scheduled.
We treat the mediator just like the other agents as far as
scheduling is concerned.
In the  sequel,}
in the sequel, we talk about an agent's (or the mediator's) \emph{turn}.
An agent's $k$th turn takes place the $k$th time it is scheduled.  During
its turn, the agent sends a block of messages and performs some local
computation.

It is more standard in the literature to assume that agents perform
at most one action when they are scheduled.  We can view this a
constraint on agents' protocols.  A \emph{single-action} protocol for
agent $i$ is one where agent $i$ sends at most one message before
performing the $\Done(i)$ action.  As we show in
Section~\ref{sec:other-models}, we could have restricted to
single-action protocols with no loss of generality as far as our
results go; allowing agents to perform a sequence of actions
atomically just makes the exposition easier.


Even though it might appear that malicious agents and the
scheduler act independently, we show in our companion paper
\cite[Section A.1]{ADGH19} that
we can assume without loss of generality that they coordinate their
actions (i.e., that they are all under the control of a single entity,
which we take here to be the scheduler).

\begin{definition}\label{dfn:adversary} An \emph{adversary} is a
  triple  $(T,\vec{\sigma}_T, 
\sigma_e)$, consisting
of a set $T$ of malicious agents,
the protocol $\vec{\tau}_T$ used by the agents in $T$,
and a protocol $\sigma_e$ for the scheduler.
An adversary where the scheduler is relaxed is a \emph{relaxed adversary}.
\end{definition}


In this paper, we consider protocols that
involve 
 a \emph{mediator}, typically denoted $d$, using a protocol denoted
$\pi_d$. In protocols that involve a mediator,
we assume that 
honest agents' strategies are always such that the honest agents
communicate only with the mediator, not with each other.   
However, since malicious agents can deviate, they can communicate
with each other. 
As far as the scheduler is concerned, the mediator is like any other
agent, so the scheduler (and the mediator's protocol) determine when
the mediator sends and receives messages.
However, the mediator is never malicious, and thus never deviates from
its  announced protocol.

We deal only with bounded protocols, where there is a bound $N$ on the
number of messages that an honest agent sends.  Of course, there is
nothing to prevent malicious agents from spamming the mediator and
sending an arbitrary number of messages.   We assume that the mediator
reads at most $N$ messages from each agent $i$, ignoring any further
messages sent by $i$.

For our results involving termination, specifically, (P2), it is
critical that players know when the mediator stops sending messages.  
For these results, we restrict the honest agents and
the mediator  to using protocols that have the following \emph{canonical form}:
\commentout{
Using a canonical protocol, agent $i$ sends a message to the mediator
in response to a message from the mediator that does not include
``STOP'' if it has not halted, and these are the only messages that
$i$ sends, in addition to an initial message to the mediator.
If agent $i$ gets a message from the mediator that includes ``STOP'',
then $i$ halts.
}
Using a canonical protocol, each honest agent tags its $\ell$th
message with label $\ell$ and all honest agents are guaranteed to
send at most $N$ messages regardless of their inputs or the random
bits they use.  Whenever the mediator receives a message from an agent
$i$, it checks its tag $\ell$; if $\ell > N$ or if the mediator has
already received a message from $i$ with tag $\ell$, it ignores the
message. The mediator is guaranteed to eventually terminate. Whenever
this happens, it sends a special ``STOP'' message to all agents and
halts. Whenever an honest agent receives a ``STOP'' message, it
terminates.  


\commentout{
However, we do require that 
if an
agent sends several messages when it is scheduled, then either all or
none of them are delivered.
}

Even though canonical protocols have a bound $N$ on the number of
messages that honest agents and the mediator can send, the mediator's local
history in a canonical protocol can be arbitrarily long, since it can
be scheduled an arbitrary number of times.
We 
conjecture 
that, in general, since the
message space is finite,   
the expected number of messages required to simulate the mediator is
unbounded. However,
we can
do better if the
mediator's protocol satisfies 
two
additional 
properties. Roughly speaking,
the first property says
that the mediator can
send messages only either at its first turn or in response to an
agent's message;
the second property says that the mediator ignores \emph{empty turns},
that is, turns where it does
not receive or send messages.  Thus, the second property implies that
the mediator cannot send a message after receiving a message that
describes how many empty turns there have been since the last time the
mediator sent a message.
More precisely, 
the first property says that whenever the mediator $\pi_d$ is scheduled
 with history $h_d$,
then if $h_d \ne (\, )$ (i.e., if $h_d$ is not the initial history) or
if the mediator has not received any messages in $h_d$ since the last
time it was scheduled, then $\pi_d(h_d) = \Done(d)$. 
The second property says that $\pi_d(h_d) = \pi_d(h_d')$, where
$h_d'$ is the result of 
removing consecutive 
($\Done(d)$, $\Sch(d)$)
 pairs in $h_d$ 
(e.g., if $h_d = (\Sch(d), \Snd(\mu,j,d),\Done(d), \Sch(d), \Done(d),
  \Rec(\mu',i,d), \Sch(d),\Done(d),\Sch(d))$, then $h_d' = (\Sch(d),
    \Snd(\mu,j,d),\Done(d), \Rec(\mu',i,d), \Sch(d))$). 
\commentout{
replacing consecutive $\Done(d)$ events in
$h_d$ by a single $\Done(d)$ event (e.g., if $h_d =
(\Snd(\mu,j,d),\Done(d), \Done(d), \Done(d),
\Rec(\mu',i,d),\Done(d),\Done(d),Done(d), \Snd(\mu',j',d))$, then
$h_d' = (\Snd(\mu,j,d), \Done(d),
\Rec(\mu',i,d),\Done(d),\Snd(\mu',j'd))$.
}
\commentout{
In Section~\ref{sec:message-bound}, we show that if the mediator uses a
responsive protocol $\pi_d$, then  we can simulate any protocol
$\vec{\pi} + \pi_d$ in such a way 
that the expected number of messages sent by honest players during the
simulation is polynomial in $n$ and $N$.
}
%
In Section~\ref{sec:message-bound}, we show that if the mediator uses a
responsive protocol $\pi_d$ that can be represented using a curcuit
with $c$ gates, then  we can simulate any protocol
$\vec{\pi} + \pi_d$ in such a way 
that the expected number of messages sent by honest players during the
simulation is polynomial in $n$ and $N$
and linear in $c$.

\section{Secure Computation in Interactive Settings}\label{sec:secure-computation} 

\subsection{The BGW/BCG notion of secure computation}


\emph{Secure computation} is concerned with jointly
computing a function $f$ on $n$ variables,  where the $i$th input
is known only to agent $i$. 
%
For instance, if we want to compute the average salary of the people
from the state of New York, then $n$ would be New York's population,
the input $x_i$ is $i$'s salary, and $f(x_1, \ldots, x_n) =
\frac{\sum_{i = 1}^n x_i}{\sum_{x_i \not = 0} 1}$.
(For the denominator we count only people who are actually working.)
Ideally, a secure
computation protocol that computes $f$ would be a protocol in which
each agent $i$ outputs $f(x_1, \ldots, x_n)$ and gains no information
about the inputs $x_j$ for $j \ne i$. In our example, this amounts to
not learning other people's salaries.

Typically, we are interested in performing secure computation in a
setting where some of the agents might be malicious and
not follow the protocol. In particular, they might not give any
information about their input or might just pretend that they have a
different input
(for instance, they can lie about
their salary). What output do we want the secure computation of $f$ to
produce in this case?  
To make precise what we want, we use notation
introduced by 
BGW and
BCG.

Let $\vec{x}$ be a vector of $n$ components; let $C$ be a subset
of $[n]$ (where we use the notation $[n]$ to denote the set
$\{1,\ldots, n\}$, as is standard); let $\vec{x}_C$ denote the vector
obtained by projecting 
$\vec{x}$ onto the indices of $C$; and if $\vec{z}$ is a vector
of length $|C|$, let $\vec{x}/_{(C, \vec{z})}$ denote the vector
obtained by replacing the entries of $x$ indexed by $C$ with $\vec{z}$. 
Given a set $C$ of indices, a default value, which we take here to
be 0, and a function $f$, we take $f_C$ to be the function 
results from applying $f$, but taking the inputs of the agents 
not in $C$ to be 0; that is,
$f_C(\vec{x}) =
f(\vec{x}/_{(\overline{C}, \vec{0})})$.
Roughly speaking, if only the agents in $C$ provide inputs, we want
the output of the secure computation to be $f_c(\vec{x})$.  

What about agents who lie about their inputs?  A malicious agent $i$
who lies about his
input $x_i$ and pretends to have some other input $y_i$ is
indistinguishable from an honest agent who has $y_i$ as his actual
input. We can capture this lie using a function 
$L:D^{|T|}
\rightarrow D^{|T|}$, where $D$ is the domain of the inputs and $T$
is the set of malicious agents. The function $L$ encodes the
inputs malicious agents pretend to have given their actual inputs.
BCG require that all the honest agent 
output the same value and that the output  has the form $(C,
f_C(\vec{y}))$, where 
$\vec{y} = \vec{x}/_{(T, L(\vec{x}_T))}$.
They allow $C$ to depend on $\vec{x}_T$, since malicious agents can
influence the choice of $C$.
They also allow the choice of $C$ and the function $L$ to be randomized.
Since the choice of $L$ and $C$
can be correlated, $L$ and $C$ are assumed to take as input a common random
value $r \in \mathcal{R}$, where $\mathcal{R}$ denotes the domain of
random inputs.
That is, $C = c(\vec{x}_T,r)$ for some function $c$, and the malicious
agents with actual input $\vec{x}_T$ pretend that their input is
$L(\vec{x}_T,r)$.

BCG place no requirements on the output of malicious agents, but they do
want the inputs of
honest agents to remain as secret as possible. Hence, in an
ideal scenario, the outputs of malicious agents can depend only on
$\vec{x}_T$, $f_C(\vec{y})$, and possibly some randomization. Taking
$O_i$ to denote the output function of a malicious agent $i$, we can
now give BCG's  
definitions.


\begin{definition}
An \emph{ideal $t$-adversary} $A$ 
is a tuple 
$(T,c,L,O)$
 consisting
of a set $T \subseteq [n]$
of malicious agents with $|T| \le t$ and three randomized
functions $c: D^{|T|} \times \mathcal{R} \rightarrow \mathcal{P}([n])$
with $|c(\vec{z},r)| 
\ge n-t$ for all input profiles $\vec{z}$ and $r$, 
$L: D^{|T|} \times \mathcal{R} \rightarrow D^{|T|}$
and $\vec{O}: D^{|T|} \times D \times \mathcal{R} \rightarrow
(\{0,1\}^*)^{|T|}$.
The \emph{ideal} output $\vec{\rho}$ of $A$
given function $f$, input profile $\vec{x}$, and a value $r \in
\mathcal{R}$ is  
\commentout{
$$\rho_i(\vec{x}, A, r; f) 
= \left
\{ \begin{array}{ll} (c(\vec{x}_T, r),f_{c(\vec{x}_T,
      r)}(\vec{x}/_{(T, h(\vec{x}_T, r))},r)) & \mbox{if } i \not \in T 
  \\ O_i(\vec{x}_T, f_{c(\vec{x}_T, r)}(\vec{x}/_{(T, h(\vec{x}_T,
        r))}), r) & \mbox{if } i \in T. \end{array}\right.$$ 
}
$$\rho_i(\vec{x}, A, r; f) 
= \left
\{ \begin{array}{ll} (c(\vec{x}_T, r),f_{c(\vec{x}_T,
    r)}(\vec{x}/_{(T, L(\vec{x}_T, r))})) & \mbox{if } i \not \in T 
    \\ O_i(\vec{x}_T, f_{c(\vec{x}_T, r)}(\vec{x}/_{(T, L(\vec{x}_T,
   r))}), r) & \mbox{if } i \in T. \end{array}\right.$$
\end{definition}
%
Let $\vec{\rho}(\vec{x}, A; f)$ denote the
distribution induced over outputs by the protocol profile
$\vec{\rho}$  on input $x$ given the ideal $t$-adversary $A$.
Note that an ideal $t$-adversary is somewhat different from the
adversary as defined in Definition~\ref{dfn:adversary}, although they
are related, 
as we show in Section~\ref{sec:secure-computation-and-mediators}.
We use variants of $A$ to denote both types of adversary.

We can now give the BCG definition of secure
computation.
Let $\vec{\pi}(\vec{x}, A)$ be the
distribution of outputs when running protocol $\vec{\pi}$ on input $\vec{x}$
with adversary $A=(T,\vec{\tau}_T,\sigma_e)$.

\begin{definition}[Secure computation]
Let $f: \mathcal{D}^n \rightarrow \mathcal{D}$ be a function on $n$
variables and $\vec{\pi}$ a protocol for $n$ agents.
Protocol
$\vec{\pi}$ $t$-securely computes $f$ if, for every
adversary $A' = (T,\vec{\tau}_T, \sigma_e)$,
the following properties hold: 
%
\begin{itemize}
\item[SC1.] For all input profiles $\vec{x}$, all honest agents terminate
  with probability 1. 
  \item[SC2.] There exists 
  an ideal 
   $t$-adversary 
    $A' = (T,c,L, \vec{O})$ 
   such that, for all input profiles $\vec{x}$,
         $\vec{\rho}(\vec{x}, A'; f)$ and $\vec{\pi}(\vec{x},A)$
are identically distributed. 
\end{itemize}
\end{definition}

Note that BCG just require that \emph{some} 
ideal 
$t$-adversary $A$ 
gives the same distribution over the the outputs of $\pi$. This
captures the idea that all 
ways that malicious agents can deviate are modeled by adversaries.
Also note that SC1 follows from SC2 if we view non-termination as
a special kind of output.

BCG prove the following 
result:

\begin{theorem}\label{thm:secure-computation}
Given $n$ and $t$ such that $n > 4t$ and a function $f:D^n \rightarrow
D$, there exists a protocol $\vec{\pi}^{f}$ that $t$-securely computes $f$. 
\end{theorem}

The construction of $\vec{\pi}^f$, which is sketched in the next
section, is of critical importance for this paper, since most of the
primitives used in this construction are also used in ours.

\subsection{The BCG construction}\label{sec:tools}
To explain the BCG construction, we must first review the tools used
by BCG, specifically,
\emph{broadcast}, \emph{consensus}, \emph{verifiable secret
sharing (VSS)}, \emph{circuit computation},  \emph{accumulative sets},
\emph{agreement on a core set},  and \emph{random polynomial generation}. 
(Accumulative sets and agreement on a core set were introduced by BCG;
the other tools are older.)

 \subsubsection{Broadcast}
 A broadcast protocol involves a \emph{sender} who sends a message $\mu$ to
all agents in such a way that all honest
agents receive the same message.
(Although we talk about ``a broadcast protocol'', this is
really a joint protocol, that is, a protocol for each agent.
Given a joint protocol $\vec{P}$, we use $P_i$ to denote $i$'s part of
the protocol.  The
  sender's protocol is different from those of the other agents.  The
  sender has input $\mu$, the message to be shared; the other agents
  have no input.)
Moreover, if the sender is honest,
the message received by an agent $i$ must be the message $\mu$ that
the sender sent.
More precisely, a \emph{broadcast} protocol invoked by a
sender with input $\mu$ must satisfy the following properties in all
histories:  
\begin{itemize}
  \item If an honest agent terminates \emph{broadcast} with output
        $\mu'$, then all honest agents  
    eventually terminate \emph{broadcast} with output $\mu'$.
      \item If the sender is honest, then all honest agents eventually
                terminate \emph{broadcast} and output $\mu$.
\end{itemize}
Bracha \citeyear{Br84} provides a broadcast protocol that 
tolerates up to $t$ malicious agents
in asynchronous systems
if $n > 3t$. 

\subsubsection{Consensus}
In a consensus protocol, each agent $i$ starts with an
initial preference $y_i \in \{0,1\}$ and must output a value $x \in
\{0,1\}$ such that the following properties are satisfied in all histories:
\begin{itemize}
\item All honest agents terminate with probability 1.
  \item If one honest agent terminates and outputs $x$, then all
        honest agents terminate and output $x$.
  \item If all honest agents have the same initial value $y$, then if
    an honest agent $i$ terminates the protocol, $i$ outputs $y$.
\end{itemize}

Abraham, Dolev and Halpern \citeyear{ADH08} provide a consensus
protocol that is $t$-resilient in asynchronous systems if $n > 3t$. 

\subsubsection{Verifiable secret sharing}\label{sec:VSS}
In a verifiable secret sharing  protocol, a sender starts out with
some secret $s$ that it wants to share.
VSS consists of a pair of
protocols $(\longvec{\VSS}^{sh},\longvec{\VSS}^{rec})$, commonly
referred to as the \emph{sharing protocol} 
and the \emph{reconstruction protocol},
and a designated agent, the \emph{sender},
such that the following properties hold:
\begin{itemize}
  \item If the sender is honest, then every honest agent $i$ will
    eventually complete $\VSS_i^{sh}$.
    \item If an honest agent  $i$ completes $\VSS_i^{sh}$, then all
  honest agents $j$  eventually complete $\VSS_j^{sh}$ and
  $\VSS_j^{rec}$. 
\item The output of $\VSS^{sh}_i$ is called $i$'s \emph{share}
  of the secret.  There is a unique value $s'$ such that if each honest
  agent $i$ runs $\VSS^{rec}_i$ with input $i$'s share of the secret,
  then all the honest agents $j$ will complete $\VSS_j^{rec}$, and 
  will output the same value $s'$, no matter what the malicious agents do.
  \item If the sender is honest, then $s' = s$ (the sender's secret).
  \item If the sender is honest and no honest agent
    $i$ has begun executing $\VSS^{rec}_i$, then the malicious agents
    cannot guess $s$ with probability $> 1/M$ (where $M$ is the
cardinality of the space of possible secrets).
\end{itemize}

With VSS, just as with the broadcast protocol, the
sender's protocol is 
different from that of the other agents; only the sender has the
secret $s$.
Whenever a recipient $i$ receives a message $\mu$ from the sender, it
invokes $\VSS_i^{sh}$ with input $\mu$
and outputs its share of the secret,
which becomes the input to $\VSS_i^{rec}$.
Even though we require each agent $i$ to output the same value $s'$ 
after runing $\VSS_i^{rec}$, a simple modification of
$\longvec{\VSS}^{rec}$ allows 
a single agent to learn the secret, without any other agent getting
any additional information:  If we want only $i$ to learn the secret,
all the agents send their shares to $i$, and $i$ simulates the
computation of $\longvec{\VSS}^{rec}$ locally.  (This depends on the
assumption that the only input to $\VSS^{rec}_j$ is $j$'s share of
the secret, and that it suffices for $i$ to learn the shares of the
honest agents in order to recover the secret.)
However, no other agents learn anything about the
secret (since all they have is their share of the secret). 

BCG provide a
VSS protocol in an asynchronous setting that is $t$ resilient as long  
as $n > 4t$.  BKR showed that if
$n > 3t$, then for all $\epsilon > 0$, there exists 
a $t$-resilient protocol that achieves the VSS properties in
asynchronous systems with probability at least $1-\epsilon$.
More precisely, their protocol has
the property that if some honest agent terminates, then all honest
agents terminate and all the properties above hold, and some honest
agent terminates with probability at least $1-\epsilon$.


\subsubsection{Accumulative sets}\label{ssection:accumulative-sets}

Suppose that we have a global clock, initialized to 0.  We do
not assume that agents have access to the global clock.
An \emph{accumulative set} is a function $U(h,m)$ from histories and global time
to sets such that $U(h,m) \supseteq U(h,m')$ if $m \ge m'$.
(Intuitively, $U(h,m)$ consists of the elements of $U$ at time $m$ in
history $h$.)

\begin{definition}
    Given  $M_1,M_2 \in \mathbb{N}$ with $M_1 \le M_2$, a tuple $(U_1,
  \ldots, U_n)$ of 
accumulative subsets of $\mathbb{N}$ (one for each agent) is
\emph{$(M_1,M_2)$-uniform} in history $h$ if, for 
every agent $i$ that is honest in $h$,
\begin{itemize}
  \item $U_i(h,m) \subseteq \{1,\ldots, M_2\}$ for all times $m \ge 0$;
      \item there exists a time $m_i^h$ such that  $|U_i(h,m_i^h)| \ge M_1$;
\item for all agents $j$ that are honest in $h$, there exists a time 
  $m_{i,j}^h$
    such that $U_i(h,m) = U_j(h,m)$ for all $m \ge m_{i,j}^h$. 
\end{itemize}
\end{definition}

To see how $(M_1,M_2)$-uniform accumulative sets are used,
suppose that each agent $i$ in a system
of $n$ agents has a secret 
$s_i$. The $n$ agents each invoke $t$-resilient VSS concurrently
in 
a system with $t$ malicious agents and $n > 3t$, with agent $i$
acting as the sender with secret $s_i$ in its invocation of VSS.  Let
$U_i(h,m)$ consist of those 
agents $j$ for which $i$ has terminated the sharing phase of the VSS
initiated by $j$ by time $m$ in history $h$.  Clearly $U_i$ is an
accumulative set. 
We claim that $(U_1, 
\ldots, U_n)$ is $(n-t,n)$-uniform.  
Clearly, $U_i(h,m) \subseteq \{1,\ldots, n\}$ for
all times $m$ by 
construction.  Since there at most $t$ malicious agents in each history
and the VSS scheme is $t$-resilient, the properties of VSS guarantee that 
each honest agent $i$ will eventually complete the VSS initiated by each 
honest agent $j$, which means $j$ is included in $U_i(h,m)$ for some $m$, 
and thus there must exist a time $m_i^h$ such that
$|U_i(h,m_i^h)| \ge n-t$. 
Since $U_i(h,m)$ is finite, there must come a time $(m_i^h)^*$ such
that $U_i(h,m')  = U_i(h,(m_i^h)^*)$ for all $m' \ge (m_i^h)^*$.
Let $m_{i,j}^h = \max((m_i^h)^*,(m_j^h)^*)$.  The properties of VSS
guarantee that $j' \in U_i(h,m_{i,j}^h)$ iff $j' \in U_j(h,m_{i,j}^h)$.

\subsubsection{Agreement on a core set}

An agreement on a core set (ACS) protocol is given as
input natural numbers $M_1$ and $M_2$.  Each agent $i$ is also
assumed to have access to an accumulative set $U_i$.  If the tuple
$(U_1, \ldots,U_n)$ is  $(M_1,M_2)$-uniform with respect to the
histories of the ACS protocol, then the following properties must hold:
\begin{itemize}
 \item All honest agents must eventually complete the ACS protocol.
  \item If an honest agent $i$ completes the protocol at time $m$, then it
    output a set $C_i \subseteq U_i(m)$ such that $|C_i| \ge M_1$.
  \item If $i$ and $j$ are honest, then $C_i = C_j$.
\end{itemize}
Thus, all honest agents running an ACS protocol must output the same
set; this set is called the \emph{core set}.  
We denote by $ACS_i(U_i, M_1, M_2)$ agent $i$'s invocation of the ACS
protocol with inputs $M_1$ and $M_2$ relative to accumulative set
$U_i$. Note that although the notation suggests that $U_i$ is the
input to $ACS_i$, the protocol may actually 
check $U_i$ several times while it is running, and $U_i$ may be
different each time it is checked, since $U_i$ may updated in parallel
with $ACS_i$.

BCG provide an ACS protocol
that is $t$-resilient in asynchronous systems if $n > 3t$.

\subsubsection{Circuit computation}\label{ssection:circuit-computation}
Another key primitive that we use is circuit computation. Let
$(\longvec{\VSS}^{sh},\longvec{\VSS}^{rec})$ 
be a VSS scheme,
and let $f: \mathbf{F}_p^N \rightarrow \mathbf{F}_p$ be a circuit with
$N$ inputs consisting only of addition and
multiplication gates. Suppose that each agent $i$ has shares $x^i_1,
x^i_2, \ldots, x^i_N$ of secrets $x_1, \ldots, x_N \in \mathbf{F}_p$
respectively (where the secrets are computed using $(\longvec{\VSS}^{sh},\longvec{\VSS}^{rec})$).
A \emph{circuit computation of $f$ (relative to $(\longvec{\VSS}^{sh},\longvec{\VSS}^{rec})$}), denoted
  $CC(f)$ (we suppress the dependence on
$(\longvec{\VSS}^{sh},\longvec{\VSS}^{rec})$ from now on) has 
  the following properties.
  We assume that there is an input $x_1, \ldots, x_N$ such that 
  each agent $i$ has shares $x_1^i, \ldots, x_N^i$ of
$x_1, \ldots, x_N$.
  Agent $i$'s component of the protocol, denoted $CC_i(f)$,
is given the inputs $x^i_1, x^i_2, \ldots, x^N_i$ and computes a
single output $y_i$, such that the following properties hold:
\begin{itemize}
\item $y_i$ is $i$'s share of $f(x_1,\ldots, x_N)$ (relative to
    $(\longvec{\VSS}^{sh},\longvec{\VSS}^{rec})$). 
  \item After running $CC_j(f)$ with inputs $x_1^j,\ldots,x_N^j$ (but
    before running the reconstruction protocol $\longvec{\VSS}^{rec}$), no malicious
    agent $j$ has any information about the shares $x^i_l$ of an honest agent
  $i$, the values $x_1, \ldots, x_N$, or $f(X_1,\ldots,x_N)$ beyond
  what it had before running $CC_j(f)$,
    even if all the malicious   agents pool their information. 
\item Even after honest agents run $\longvec{\VSS}^{rec}$, no malicious agent $j$ can guess
    the values of the shares $x^i_l$ of an honest agents $i$ or the
  the secrets $x_1, \ldots, x_N$ any better than it could before
  running $CC_j(f)$ if it were given $f(x_1, \ldots, x_N)$. 
\end{itemize}
Simply put, a circuit computation protocol $CC$ allows agents
to compute their share of the output of an arithmetic circuit given
their shares of the circuit's inputs, without revealing any information.

Since it is well known that every function $f: D^N \rightarrow D$ can
be represented by a circuit $f': \mathbf{F}_p^N \rightarrow \mathbf{F}_p$ for a prime
$p \ge |D|$ (viewing the elements of $D$ as the first $|D|$ elements
of $\mathbf{F}_p$), if we can define a protocol $CC(f)$ for all
arithmetic circuits, then we can define a protocol $CC(f)$ for all
functions $f: D^N \rightarrow D$.
This is especially important in the next section, where we use CC to
compute functions whose inputs and outputs are local histories. 

BCG provide
an implementation of $CC(f)$ for all arithmetic circuits $f$ relative to
the VSS protocol that they provide 
that is $t$-resilient in asynchronous systems 
as long as $n > 4t$; 
given $\epsilon > 0$, 
BKR provide an implementation
of $CC(f)$ for all arithmetic circuits relative to the VSS protocol
that they provide that is $t$-resilient in asynchronous systems and has at most
an $\epsilon$ probability of error (i.e., there is a probability
$\epsilon$ that agents remain in deadlock or the output of the
computation will not be the 
appropriate share of the circuit's output) as long as $n > 3t$. 

We can assume without loss of generality that
CC can handle randomized functions. 
That is, if there is a protocol
$CC(f)$ to securely compute every deterministic function $f: D^N
\rightarrow D$, then there is a protocol $CC(f)$ to securely compute
every randomized function $f: D^N \rightarrow D$.  A randomized
function $f: D^N \rightarrow D$ can be viewed as a deterministic
function once it is given sufficiently many random bits,
that is, it can be identified with a deterministic function $f: D^N \times
\{0,1\}^{N'} \rightarrow D$ for $N'$ sufficiently large.  Using 
ACS, VSS, and (deterministic) CC,
 the agents can easily compute shares for $N'$ random bits as
follows. 
\begin{enumerate}
\item Each agent $i$ chooses a random bit $b_i$ and shares it using
    VSS.
\item Using ACS, the agents agree on a common set $C$ consisting of
    at least $t+1$ agents who correctly shared a bit $b_i$ at step
  1. Set $b:= \oplus_{i \in C} b_i$ (where $\oplus$ denotes sum mod 2).
  \item Each agent $i$ computes its share of $b$ using CC. 
\end{enumerate}
If $n > 2t$, then there are at least $t+1$ honest agents, so each
honest agent will get shares from at least $t+1$ agents.  Since the
set of $t+1$ agents agreed on using ACS contains at least one honest
agent, the bit $b$ must be truly random.

We can also assume without loss of generality that whenever an honest
agent terminates a CC computation of some function $f(x_1, \ldots,
x_N)$, even in the presence of at most $t$ malicious agents, at least
$n-2t$ other honest agents $i$ have computed their share $y_i$ of
$f(x_1, \ldots, x_N)$. This can be ensured by having an honest agent
$i$ send a \emph{Ready} message to all agents when it finishes the
conputation of $y_i$, and terminating the CC
procedure when it receives $n-t$ \emph{Ready} messages. If
there are at most $t$ malicious agents, if an agent receives $n-t$
\emph{Ready} messages, at least $n-2t$ are from honest agents who
genuinely computed their own share. This property will be critical
later, since it guarantees that sufficiently many honest
agents are running the protocol at roughly the same pace. 

\subsubsection{Construction of
  $\vec{\pi}^{f}$}\label{sec:construction} 

Using the primitives sketched above, BCG gave a construction of
$\vec{\pi}^f$.  At the high level, the construction proceeds as follows:
for $\vec{\pi}^{f}$:
\begin{enumerate}
  \item Each agent $i$ shares its input using VSS.
  \item Agents agree on a core set $C \subseteq [n]$ with $|C| \ge n-t$
    using an ACS procedure with parameters $M_1 = n-t$ and $M_2 = n$,
where the accumulative   
set $U_i$ of agent $i$ is the set of agents $j$ 
such that $i$ has terminated the VSS invoked by $j$ at
step 1. 
\item Each agent $i$ computes its share of $f_C(\vec{x})$ using CC, where 
$i$'s input for the $j$th input gate is $i$'s share of $x_j$ if $j \in
    C$;  otherwise it is 0. 

\item Each agent $i$ sends its share of $f_C(\vec{x})$ to each other
  agent $j$, then uses the shares received from other agents to
    reconstruct $f_C(\vec{x})$ using VSS.
\item Each agent $i$ outputs $(C, f_C(\vec{x}))$.
\end{enumerate}

\subsection{Secure computation and mediators}\label{sec:secure-computation-and-mediators}

Even though it is not explicitly proven by BCG, their construction of
$\vec{\pi}^{f}$ satisfies an additional property that we call SC3, which is
essentially a converse of SC2. 

SC3. For all
ideal 
$t$-adversaries 
$A = (T,c,L,\vec{O})$, 
there exists an adversary $A' = (T,\vec{\tau}_T, \sigma_e)$ such
that, for all input profiles $\vec{x}$, $\vec{\rho}(\vec{x}, A;f)$ and
$\vec{\pi}(\vec{x}, A')$ are identically
distributed. 

\begin{lemma}
Given a function $f: D^n \rightarrow D$, protocol $\vec{\pi}^f$ satisfies SC3.
\end{lemma}

\begin{proof}
    Suppose that $A$ is deterministic (i.e., $c$, $L$, and $\vec{O}$ do not
  depend on the random string $r$). Given
    $\vec{x}_T$, let $C = c(\vec{x}_T)$ and 
$\vec{y} = L(\vec{x}_T)$
  (note that we have dropped the $r$ input to both functions since both
  are independent of $r$). Consider the protocol $\vec{\tau}_T$ such
  that, if the agents in $T$ have input $\vec{x}_T$,  $\tau_i$ consists
of $i$ running $\pi_i$ with input $y_i$, where 
$\vec{y}_T = L(\vec{x}_T)$, except  
that 
if $i$ was supposed to output $(S,z)$ (note that all outputs of honest
agents are of this form, since the ideal output has this form, and
$\pi^f$ securely computes $f$) it outputs $O_i(\vec{x}_T, z)$ instead.
\commentout{
Finally, consider a scheduler $\sigma_e$ that delays all messages to
and from the agents in 
$\overline{C}$
and that schedules the messages from the remaining players in such a
way that all players in $C$ terminate the VSS invocations initiated by all
other players in $C$ before they agree on the core set. 
\commentout{
By the properties of
Section~\ref{sec:tools},}
Clearly, in this case, the core set computed by honest players will
be $C$;
}
Suppose that the scheduler delays messages to and from players in $\bar{C}$
until all others players terminate, and that it delivers messages in
such a way that all players in $C$ terminate all instances of VSS
invoked by a player in $C$ before they receive any other messages.
The properties of the ACS protocol guarantee that, in
this case, the core set computed by honest players will be $C$.
It follows from 
the construction of $\vec{\tau}_T$
 that 
  $\vec{\rho}(\vec{x},
A)$ and $\vec{\pi}(\vec{x}, T, \vec{\tau}_T, \sigma_e)$ are
identically distributed. If $A$ is randomized, $\tau_T$ works the
same way except that it chooses $C$, $\vec{y}_T$, and $\vec{O}_T$ by
sampling from the same distribution that $r$ is sampled from.
\end{proof}

We next show how secure computation relates to simulating a
mediator. Consider the following 
protocol 
$\medProtocol$ for $n$ players and a mediator:
Agents send their inputs to the 
mediator the first time that they are scheduled. The mediator waits until
it has received a valid input 
from all agents in a subset $C$ of agents with $|C|
\ge n-t$. The mediator then computes $y = f_C(\vec{x})$ and sends 
each agent the 
pair $(C, y)$.
When the agents receive a message from the mediator,
they output that message and terminate.

Clearly 
$\medProtocol$ 
 satisfies SC1.  It is easy to see that it also
satisfies SC2:  
Given
a set $T$ of malicious agents, a deterministic protocol profile
$\vec{\tau}_T$ for the 
malicious agents, and a deterministic scheduler $\sigma_e$, 
define 
$L(\vec{x}, r)$
 to be
whatever the malicious agents send to the mediator with
input $\vec{x}$, 
let $c(\vec{x})$ be the set of agents from whom
the mediator has received a message the first time it is scheduled after
having received a message from a least $n-t$ agents (given
$\sigma_e$, $\vec{\tau}_T$, and input $\vec{x}$), and
let $O(\vec{x})$ be the output
function that malicious agents use in 
$\medProtocol$ 
 (note that they receive
a single message with the output of the computation, so their output
depends only on $\vec{x}$, 
$\vec{\tau}_T$,
 and $\sigma_e$).
Clearly SC2 holds with this choice of $t$-ideal adversary.
Randomized functions $\vec{\tau}_T$ and $\sigma_e$ 
can be viewed as resulting from sampling random bits $r$ according to
some distribution and then running deterministically; the protocols
$c$, $h$, and $O$ can sample $r$ from the same distribution and then
proceed as above with respect to the deterministic 
$\vec{\tau}_T(r)$
 and
$\sigma_e(r)$.  

The 
 protocol 
$\medProtocol$ 
 satisfies SC3 as well.
Given 
$A = (T,c,L,O)$,
 the definition of $\vec{\tau}_T$ and $\sigma_e$
is straightforward: 
the agents in $T$ choose a
random input $r \in \mathcal{R}$ and then each agent $i \in T$ sends
$L(x_i, r)$
 to the mediator. The scheduler $\sigma_e$ delivers all
messages from the agents in $c(\vec{x}_T, r)$ first, and then
schedules the mediator.
It then delivers all the other messages.

Since both 
$\medProtocol$ and $\vec{\pi}^f$ 
 satisfy SC2 and SC3, for
  all adversaries $A$, there exists an adversary $A'$  (resp., for all
adversaries $A'$ there exists an adversary $A$) such that
$(\medProtocol)(\vec{x},A)$ and $\vec{\pi}^f(\vec{x},A')$
 are
identically distributed.

Unfortunately, given a protocol $\vec{\pi}_d$ for the mediator, there
might not exist a function $f$ such that SC2 and SC3 hold, as the
example given in the introduction (where the mediator sends to the
agents the first message it receives) shows.
Note that, in this example, the output of the
agents is not a function of their input profile,
there is no function $f$ for which SC2 and SC3 hold.
Nevertheless, we are still interested in securely computing the output
of the protocol with the mediator.
That is, 
we are interested in getting analogues to SC2 and SC3
for arbitrary 
interactive protocols. 

\begin{definition}\label{def:bisimulation}
  Protocol $\vec{\pi}'$ 
  \emph{$t$-bisimulates} 
$\vec{\pi}$ 
   if the following
two properties hold:  
\begin{itemize}
\item[(a)] 
For all adversaries $A = (T, \vec{\tau}_T, \sigma_e)$ with $|T| \le
t$, there exists an adversary $A' = (T, \vec{\tau}'_T, \sigma_e')$
such that for all input profiles $\vec{x}$, 
$\vec{\pi}(\vec{x},A)$
and $\vec{\pi}'(\vec{x},A')$ are identically distributed.
\item[(b)]
For all adversaries $A' = (T, \vec{\tau}_T', \sigma_e')$ with $|T| \le
t$, there exists an adversary $A = (T, \vec{\tau}_T, \sigma_e)$
such that all input profiles $\vec{x}$, 
$\vec{\pi}(\vec{x},A)$
and $\vec{\pi}'(\vec{x},A')$ are identically distributed.
\end{itemize}
\end{definition}

Note that the first clause is analogous to SC2, while the second
clause is analogous to SC3. 
There is no clause analogous to SC1 since we allow agents not
to terminate. In any case, since we can 
view non-termination as a special type of output (i.e., we can view an
agent that does not terminate as outputting $\bot$), so SC2 already guarantees
that non-termination happens 
with the same probability in $\vec{\pi}'$ and 
$\vec{\pi}$
(In the setting of BGW, since all functions terminate, with this viewpoint,
SC2 implies SC1, a point already made by Canetti \citeyear{canetti96studies}.)

Our earlier discussion proves the following proposition:

\begin{proposition}\label{prop:secure-comp-bisimulation}
    $\vec{\pi}^f$ $t$-bisimulates 
    $\medProtocol$ 
     if $n > 4t$.
\end{proposition}

\subsection{Beyond secure computation}\label{sec:additional}

We view $t$-bisumulation as capturing the essence of secure
computation, as defined by BCG (and others).  However, there
are two additional properties that we need for
the results of our companion paper \cite{ADGH19}, which we
believe are of independent interest.   
Both of them 
are in fact satisfied by $\pi^f$, although BCG do not discuss them.

To understand the first property, note that 
Proposition~\ref{prop:secure-comp-bisimulation}
guarantees that $\pi^f$ $t$-bisimulates 
$\medProtocol$ 
 if $n > 4t$. What
happens if $t$ 
is larger than this threshold?
Although BCG make claims for their protocol only if $n > 4t$,
variants of
some of
the properties that they are interested in continue to hold even if
$n/4 \le t < n/3$.
Specifically, for each adversary $A = (T, \vec{\tau}_T,\sigma_e)$ ,
there exists a  
relaxed adversary $A' = (T, \vec{\tau}'_T, \sigma'_e)$ 
such that, for all input profiles $\vec{x}$,
$\vec{\pi}^f(\vec{x},A)$ and 
$(\medProtocol)(\vec{x},A')$
are identically
distributed. 
This means that if $n > 3t$, then
the only way that the adversary can affect $\vec{\pi}^f$ is by
preventing some agents from terminating. This motivates the following
definition: 


\commentout{
A protocol
$\vec{\pi}'$ \emph{$(t, t')$-emulates} $\vec{\pi}$ if it $t$-emulates
$\vec{\pi}$ but the scheduler $\sigma_e$
of Definition~\ref{def:bisimulation} 
may be relaxed for 
$t \ge |T| > t'$.
}
\begin{definition}
A protocol
$\vec{\pi}'$ \emph{$(t, t')$-bisimulates} $\vec{\pi}$ if it $t$-bisimulates
$\vec{\pi}$ but the schedulers $\sigma'_e$ and $\sigma_e$ of the first and second clause 
of Definition~\ref{def:bisimulation} respectively 
may be relaxed for 
$t \ge |T| > t'$.
\end{definition}


\begin{proposition}\label{pro:emulate}
$\vec{\pi}^f$ $(t, t')$-bisimulates $\vec{\pi}_d^f$ if $3t + t' < n$
  and $t \ge t'$. 
\end{proposition}


As we just observed, if $4t \le n < 3t$, then  some honest agents might
not terminate.
However, we can show that the BCG protocol has the property that if at
least $2t+1$ honest agents terminate, then all the remaining honest
agents terminate.  
This observation motivates the following definition: 

\begin{definition}\label{def:cotermination}
  A protocol $\vec{\pi}$ \emph{$(t,k)$-coterminates} if,
all adversaries 
$A =  (T,\vec{\tau}_T, \sigma_e)$
 with $|T| \le t$ 
  and all input profiles $\vec{x}$, in 
  all histories of  
  $\vec{\pi}$ with adversary $A$ and input $\vec{x}$,
either all the agents 
not in $T$
terminate or strictly fewer than $k$ agents not in $T$ do.
\end{definition}


\begin{proposition}\label{pro:cotermination}
$\vec{\pi}^f$ $(t,2t+1)$-coterminates.
\end{proposition}

We do not prove Proposition~\ref{pro:emulate}
or~\ref{pro:cotermination} here, since we prove a generalization of
them below  
(see Theorem~\ref{thm:main}).

\subsection{Simulating arbitrary protocols}

The goal of this paper is to show that we can securely implement any
interaction with a mediator, and do so in a way that ensure the two 
properties discussed in Section~\ref{sec:additional}.
This is summarized in the
following theorem: 

\begin{theorem}\label{thm:main}
\commentout{
Given a mediator game $\Gamma_d$ with $n$ agents, for every protocol
$\vec{\pi} + \pi_d$ in $\Gamma_d$, there exists a protocol $\vec{\pi}'$
in the cheap-talk such that $\vec{\pi}'$
}
For every protocol $\vec{\pi} + \pi_d$ for $n$ agents and a mediator, there exists a protocol $\vec{\pi}'$ for $n$ agents such that $\vec{\pi}'$
\begin{itemize}
\commentout{
  \item [(a)] $t$-bisimulates $\vec{\pi}$ if $n > 4t$,
  \item [(b)] $(t,t')$-emulates $\vec{\pi}$ if  $3t + t' < n$ and $t \ge t'$, 
\item [(c)] $(t,2t+1)$-coterminates if $t < n/3$,
}
\item[(a)] $(t,t')$-bisimulates $\vec{\pi}$ if $n> 3t+t'$ and $t \ge t'$,
  \item[(b)] $(t,2t+1)$-coterminates if $n > 3t$
  and $\vec{\pi} + \pi_d$ is in canonical form.
\end{itemize}
\commentout{
The expected number of messages sent in  a history of
$\vec{\pi}'$ is 
polynomial in $n$ and $N$, and linear in $c$, 
where $N$ is the expected number of messages sent when running $\vec{\pi} + \pi_d$.
}
Moreover, if $\pi_d$ is responsive, 
the expected number of messages sent a history of
$\vec{\pi}'$ is polynomial in $n$ and $N$, and linear in $c$,  
where $N$ is the expected number of messages sent when running
$\vec{\pi} + \pi_d$
and $c$ is the number of gates in an arithmetic circuit that
implements the mediator's protocol.
\end{theorem}

The construction of $\vec{\pi}'$ is given in 
Section~\ref{sec:implementation} 
and, not surprisingly, uses many of the techniques used by BCG.
And, like BKR, if we allow an $\epsilon$ probability of error we get
stronger results. We define $\epsilon$-$t$-bisimulation
just like $t$-bisimulation (Definition~\ref{def:bisimulation}), except
that, in both clauses, 
the distance between $(\vec{\pi}+\pi_d)(\vec{x},A)$
and $\vec{\pi}'(\vec{x},A')$ is less
than $\epsilon$, where the distance $d$ between probability measures
$\nu$ and $\nu'$ on some finite space $S$ is defined as $d(\nu,\nu') =
\sum_{s \in S} |\nu(s) - \nu'(s)|$. The definition of
$\epsilon$-$t$-bisimulation and 
$\epsilon$-$(t,t')$-bisimulation are analogous. A
protocol \emph{$\epsilon$-$(t,k)$-coterminates} if it
$(t,k)$-coterminates with probability $1-\epsilon$.  
\begin{theorem}\label{thm:main-eps}
For every protocol $\vec{\pi} + \pi_d$ for $n$ agents and a mediator
and all $\epsilon > 0$, there exists a protocol $\vec{\pi}'$ for $n$
agents such that $\vec{\pi}'$ 
\begin{itemize}
  \item[(a)] $\epsilon$-$(t,t')$-bisimulates $\vec{\pi}+\pi_d$ if $n > 2t+t'$
  and $t \ge t'$, 
  \item[(b)] $\epsilon$-$(t,t+1)$-coterminates if $n > 2t$
  and $\vec{\pi} + \pi_d$ is in canonical form.
\end{itemize}
\commentout{
The expected number of messages sent in  a history of
$\vec{\pi}'$ is 
polynomial in $n$ and $N$, and linear in $c$, 
where $N$ is the expected number of messages sent when running $\vec{\pi} + \pi_d$.
}
Moreover, if $\pi_d$ is responsive, $\vec{\pi}'$ can be implemented in
such a way that the expected number of messages when running
$\vec{\pi}' + \pi_d$ is polynomial in $n$ and $N$, and linear in $c$,  
where $N$ is the expected number of messages sent when running $\vec{\pi} + \pi_d$.
\end{theorem}

\subsection{Adversaries in asynchronous systems}\label{sec:adversaries}

Even though throughout this paper we consider the scheduler and
malicious players to be separate adversarial entities, we show next
that they can coordinate. In particular, we show that malicious
players and the scheduler can communicate even though the scheduler
cannot send or receive messages. 

To see that a malicious player can send information to the scheduler,
it suffices to note that a player can encode any unary string $1^k$ by
5for instance, use the following scheme: Whenever a malicious player
sending $k$ messages to itself. Since the scheduler knows the
number of messages sent by each player and the recipient of each 
message, it can ``receive'' such messages.
patterns. The scheduler can communicate with the malicious players
by using the following scheme: Whenever a malicious player
$i$ is scheduled, it just sends a message to itself and performs a
$\Done$ action. The scheduler then delivers the message to
$i$ and schedules $i$ again. This process is repeated until $i$ is
scheduled before receiving the message it sent to itself.
Agent $i$ can interpret the number $k$ of times that this process
is repeated as the unary string $1^k$. 

Since the scheduler and the malicious players can coordinate by
communicating in this fashion, this shows that without loss of generality we
can view the adversary as a single entity that controls both the
malicious parties and the scheduler simultaneously. We will assume
such an adversary in the rest of the paper.

\subsection{Other models}\label{sec:other-models}

Before proving our main results, we discuss some of the choices made
in our formal model and show that they are essentially being made
without loss of generality.
We start by considering 
our assumption that agents perform a sequence of actions
atomically when they are scheduled.
We next 
 show
that
 we would get theorems equivalent  
 to the ones that we are claiming
if we had instead assumed that agents perform just a single action
when they are scheduled.
To prove this, we first need the following notion:

\begin{definition}
A protocol $\vec{\pi}$ is \emph{$N$-message bounded} if for all inputs
and all histories, no player ever sends more than $N$ messages in a
single turn. A protocol is \emph{message bounded} if it is
$N$-message bounded for some $N$. 
\end{definition}


\begin{proposition}\label{prop:equivalence}
  There exist a function $H$ from
    message-bounded  
   protocols to single-action
  protocols such that for all profiles $\vec{\pi}$, the following
holds: 
\begin{itemize}
\item[(a)] For all schedulers (resp., relaxed schedulers) $\sigma_e$
  there exists a scheduler (resp., relaxed scheduler) $\sigma_e'$ such
  that, for all input profiles $\vec{x}$, 
\commentout{
  $O(\vec{\pi}, \sigma_e, \vec{x})$
  and $O(H(\vec{\pi}), \sigma_e', \vec{x})$
}
$\vec{\pi}(\vec{x}, \sigma_e)$ and $H(\vec{\pi})(\vec{x}, \sigma_e')$
are identically distributed,
    where we take $H(\vec{\pi}) = (H(\pi_1), \ldots, H(\pi_n))$
  and we view $\sigma_e$ and $\sigma_e'$,  respectively, as the
  adversaries (i.e., we take $T = \emptyset$).
  
\item [(b)] For all schedulers (resp., relaxed schedulers) $\sigma_e'$
  there exists a scheduler (resp., relaxed scheduler) $\sigma_e$ such
  that, for all input profiles $\vec{x}$, 
\commentout{  
  $O(\vec{\pi}, \sigma_e, \vec{x})$
and $O(H(\vec{\pi}), \sigma_e', \vec{x})$
}
$\vec{\pi}(\vec{x}, \sigma_e)$ and $H(\vec{\pi})(\vec{x}, \sigma_e')$
 are identically distributed.
\end{itemize}
\end{proposition}

The converse of Proposition~\ref{prop:equivalence} is trivial, since
single-action strategies are strategies.  

It
follows from
Proposition~\ref{prop:equivalence} that Theorem~\ref{thm:main} holds
even if we restrict agents to using single-action strategies.
\commentout{
To see this, note that given a single-action protocol with a mediator
$\vec{\pi} + \pi_d$, we can view it as a standard protocol (in our
setup), and construct a protocol $\vec{\pi}'$ without a mediator that
satisfies the properties of Theorem~\ref{thm:main}. Now, consider the
single-action protocol $\vec{\pi}'' := H\left(\vec{\pi}'\right)$. 
}

\begin{proof}
%
Intuitively, $H(\pi_i)$ is identical to $\pi_i$, except that rather
than sending a sequence of messages when it is scheduled, $i$ sends
the messages one at a time.  The scheduler $\sigma_e'$ is then chosen to ensure
that $i$ is scheduled so that it sends all of its messages as if they
were sent atomically.  In addition to keeping track of the messages it
has sent and received, $i$ uses the variable $U_i$ whose value is a
sequence of mesages (intuitively, the ones that $i$ would have sent at
this point in the simulation of $\pi_i$ that it has not yet sent), initally set 
to the empty sequence, and a binary variable $\next$, originally set
to 1.
 When $i$ is scheduled by $\sigma_e'$, $H(\pi_i)$ proceeds as follows:
 If $\next = 1$, then $i$ sets $U_i$ to the sequence of
 messages that it would send with $\pi_i$ given its current history.
 (If $\pi_i$ randomizes, then $H(\pi_i)$ does the same randomization.
If $U_i$ is the empty sequence (so $\pi_i$ would not send any
messages at that point), $i$ performs the action $\Done(i)$, and
 outputs whatever it does with $\pi$;  otherwise, $i$ sets $\next$ to
 $0$, sends the first message in  $U_i$ to its intended recipient, and
 removes this message from $U_i$.  If $\next=0$, then if $U_i$ is
 empty, $i$ sets $\next$ to 1, sends  $\Done(i)$, and outputs whatever
 it does with $\pi$; otherwise,  $i$  sends the first  
message in $U_i$ to its intended recipient and removes it from $U_i$.

Since $\vec{\pi}$ is message bounded, there exists an $N$ such that
$\vec{\pi}$ is $N$-message bounded. For part (a),  
given 
$\sigma_e$, we construct
$\sigma_e'$ so that it simulates $\sigma_e$, except that 
if $\sigma_e$ schedules $i$, $\sigma_e'$ schedules $i$ repeatedly
until 
either 
it observes $\Done(i)$
or until $i$ sends messages in $N+1$ consecutive turns.
Since $\vec{\pi}$ is $N$-message bounded, it
is clear that
$\vec{\pi} (\vec{x}, \sigma_e)$ and $H(\vec{\pi})(\vec{x}, \sigma_e')$
 are identically
distributed.
Note that it is necessary for $\vec{\pi}$ to be $N$-message bounded,
since if the scheduler schedules each player $i$ repeatedly until it
stops sending a message during its turn, a player that keeps sending
messages would be scheduled indefinitely, and so would prevent other
players from being scheduled.

For part (b), given $\sigma_e'$, we construct $\sigma_e$ so that it
simulates $\sigma_e$.  There is one issue that we have to deal with.
Whereas with $\sigma_e$, an agent $i$ can send $k$ messages each time
it is scheduled, with $\sigma_e'$, it can send only one message when it is
scheduled.  The scheduler $\sigma_e'$ constructed from $\sigma_e$ in
part (a) scheduled $i$ repeatedly until it sent all the messages it
did with $\sigma_e$.  But we cannot assume that the scheduler
$\sigma_e'$ that we are given for part (b) does this.  Thus,
$\sigma_e$ must keep track of how many of the messages that each agent $i$
was supposed to send the last time it was scheduled by $\sigma_e$ have
been sent so far.  To do this,  $\sigma_e$ uses variables
$mes_i$, one for each agent $i$, initially set to 0, such that
$mes_i$  keeps track of how many of the messages that agent $i$ sent with
$\sigma_e$ still need to be sent by $\sigma_e'$. 
As we observed above, given a local history $h$ of 
the scheduler where the agents use 
$\vec{\pi}$ 
 and
the scheduler uses 
$\sigma_e$, there is a corresponding 
local history $h'$ of the scheduler where the agents use $\vec{\pi}'$
and the scheduler uses  $\sigma_e'$.   If, given $h'$, $\sigma_e'$
schedules agent $i$ with  probability $\alpha_i$, then with the same
probability $\alpha_i$, $\sigma_e$ proceeds as follows: if $mes_i =
0$ (which means that all the messages that $i$ sent the last time
it was scheduled have been delivered in $h'$), then $\sigma_e$
schedules $i$, sees how many messages $i$ delivers according $\pi_i$,
and sets $mes_i$ to this number; if $mes_i \ne 0$, then $mes_i$ is 
decremented by 1 but no agent is scheduled.
Again,
it is clear that
that 
$\vec{\pi}(\vec{x}, \sigma_e)$ and $H(\vec{\pi})(\vec{x}, \sigma_e')$ 
 are identically
distributed.
\commentout{
For part (b), given $\sigma_e'$, the construction of $\sigma_e$ is
straightforward.  As we observed above, given a local history $h$ of 
$\sigma_e$ where the agents use $\vec{\pi}$, there is a corresponding
local history $h'$ of $\sigma_e'$.  If, 
in history $h$, agent $i$ made the last move and did not perform
$\Done(i)$, then $\sigma_e$ does nothing (agent $i$ is in the midst
of an atomic move).  If the last move in $h$ has the form $\Done(i)$
for some agent $i$ or was a move by the schduler, then $\sigma_e$
had the same distribution of moves in $h$ as $\sigma_e'$ does in
$h'$.  Again, it is clear that the histories obtained with
$H(\vec{\pi})$ and $\sigma_e'$ are identically distributed to the
corresponding histories with $\vec{\pi}$ and $\sigma_e$, and hence
that $O(\vec{\pi}, \sigma_e, 
\vec{x})$ and $(H(\vec{\pi}), \sigma'_e, \vec{x})$. are identically
distributed.
}
\end{proof}

BCG put further constraints on the scheduler.  Specifically, they
assume that,
except possibly for the first time that agent $i$ is scheduled, 
$i$ is scheduled immediately after receiving a message and only then.
That is, in our terminology, BCG assume that a $\Rec(\cdot, \cdot, i)$
event must be followed by a $\Sch(i)$ event,
and all $\Sch(i)$ events except possibly the first one occur after a
$\Rec(\cdot,\cdot,i)$ event.
We call 
the schedulers that satisfy this constraint \emph{BCG schedulers}.

We now prove a result analogous to Proposition~\ref{prop:equivalence},
from which it follows that we could have obtained our results using a 
BCG scheduler.

\begin{proposition}\label{prop:equivalence2}
  There exist a function $H$ from strategies to strategies such that for
all strategies $\vec{\pi}$ the following holds: 
\begin{itemize}
\item[(a)] For all
  schedulers (resp., relaxed schedulers) $\sigma_e$
  there exists a BCG scheduler (resp., relaxed BCG scheduler)
  $\sigma_e'$ such that, for all input profiles $\vec{x}$, 
\commentout{  
  $O(\vec{\pi} +\pi_d,
  \sigma_e, \vec{x})$ and  $O(H(\vec{\pi}),\sigma_e', \vec{x})$
  }
$\vec{\pi}(\vec{x}, \sigma_e)$ and $H(\vec{\pi})(\vec{x}, \sigma_e')$
  are identically distributed.

\item [(b)] For all BCG schedulers (resp., relaxed schedulers)
  $\sigma_e'$ there exists a scheduler (resp., relaxed scheduler)
  $\sigma_e$ such that, for all input profiles $\vec{x}$, 
\commentout{  
  $O(\vec{\pi}+\pi_d,
  \sigma_e, \vec{x})$ and  $O(H(\vec{\pi}),\sigma_e', \vec{x})$
  }
  $\vec{\pi}(\vec{x}, \sigma_e)$ and $H(\vec{\pi})(\vec{x}, \sigma_e')$
  are identically distributed.
\end{itemize}
\end{proposition}

\begin{proof}
As in Proposition~\ref{prop:equivalence}, the idea is that $\sigma_e'$
simulates $\sigma_e$, but since $\sigma_e$ can schedule an agent only
when it delivers a 
message, we have each agent $i$ send itself special messages, denoted
$\proceed_i$, to ensure that there are always enough messages in the
system.  In more detail, $H(\pi_i)$ works as follows.  When it is first
scheduled, agent $i$ sends itself a $\proceed_i$ message.  Since we
are considering BCG schedulers, agent $i$
is scheduled subsequently only when it receives a message.  If it
receives a message other than $\proceed_i$, it does nothing (although
the message is added to its history).  If it receives a $\proceed_i$
message, then it does whatever it would do with $\pi_i$ given its
current history with the $\proceed_i$ messages
and the $\Sch(i)$ events not preceeded by a $\proceed_i$ message 
 removed, and sends
itself another $\proceed_i$ message.

For part (a), given $\sigma_e$, $\sigma_e'$ first schedules each
agent once (in some arbitrary order), to ensure that that each of
them has sent a $\proceed_i$ message that is available to be
delivered.  Given a history $h'$, $\sigma_e'$ considers
what $\sigma_e$ would do in the history $h$ that results from $h'$ by
removing 
the initial $\Sch(i)$ event for each agent $i$, the last message
that each agent $i$ sends when it is scheduled if it sends a message
at all, and the receipt of these messages.  If $h'$ is a history that
results where the agents are running $H(\vec{\pi})$, then the send
and receive events removed are precisely those that involve $\proceed_i$. 
If $\sigma_e$ delivers a message with some probability, then
$\sigma_e'$ delivers the corresponding message with the same
probability; if $\sigma_e$ schedules an agent $i$ with some
probability, $\sigma_e'$ delivers the last $\proceed_i$ that $i$ sent
and schedules agent $i$ with the same probability.  If there is no
$\proceed_i$ message to deliver, then $\sigma_e'$ does nothing, but
our construction of $H(\pi_i)$ guarantees that if $h'$ is a
history that results from running $H(\vec{\pi})$, 
then there will be such a message that can be delivered.
Again, it is clear that
$\vec{\pi}(\vec{x}, \sigma_e)$ and $H(\vec{\pi})(\vec{x}, \sigma_e')$
 are identically
distributed.

For part (b), given $\sigma_e'$, the construction of $\sigma_e$ is
similar to that of Proposition~\ref{prop:equivalence}.  Again, given a
local history $h$ of $\sigma_e$ where the agents use 
$\vec{\pi}$,
there is a corresponding history $h'$ of $\sigma_e'$ where the agents
use $H(\vec{\pi})$.
If, given input $h'$, $\sigma_e'$ delivers a message with some
probability $p$ and the messages is not a $\proceed_i$ message, then
$\sigma_e$ delivers the corresponding message with  
probability $p$.
If the message is a $\proceed_i$
message, then $\sigma_e$ also schedules agent $i$.  If 
$\sigma_e'$ schedules an agent $i$ with 
probability $p$, and in $h'$ this is the first time that $i$ is
scheduled, then $\sigma_e$ schedules $i$ with probability $p$ and 
otherwise does nothing with probability $p$.  
 Yet again,
it is straightforward to show that
$\vec{\pi}(\vec{x}, \sigma_e)$ and $H(\vec{\pi})(\vec{x}, \sigma_e')$ 
  are identically
distributed.
\end{proof}

\section{The Proof of Theorems~\ref{thm:main} and~\ref{thm:main-eps}}
In this section, we prove Theorems~\ref{thm:main}
and~\ref{thm:main-eps}.  Since the proofs are 
rather complicated, we proceed in stages. 
\subsection{$t$-uniform VSS and CC and determinate VSS}\label{sec:uniform-CC}
BCG's implementation of VSS satisfies some additional properties that
they do not make use of, but that we will need in our construction, so
we outline them here.

Given a sequence $I = \{i_1, \ldots, i_n\}$ of distinct honest agents,
a sequence $S = \{s_1, \ldots, s_n\}$ of values, and a secret $s$,
we say that $(I,S)$ is \emph{$s$-realizable} by a VSS (resp., CC)
implementation if, for that implementation, there exists an agent $i$
such that the  
the event that each agent $i_k$   
computes $s_k$ as the output of
$i$'s invocation of VSS with secret $s$ has nonzero probability. In other
words, $(I,S)$ is $s$-realizable if $S$ could be the output of the agents
in $I$ running VSS with secret $s$.
$(I,S)$ is realizable if it is $s$-realizable for some $s$. 

We say that $(I', S')$ is an
$s$-\emph{extension} of $(I,S)$ if $I$ is a prefix of $I'$, $S$ is a
prefix of $S'$
and $(I', S')$ is $s$-realizable. $(I', S')$ is
a \emph{full} $s$-extension if $I'$ is the set of all agents. Again, we say
that $(I',S')$ is an extension of $(I,S)$ if it is an $s$-extension of
$(I,S)$ for some
$s$; it is a \emph{full extension} if, in addition, $I'$ is the set of all
agents. 
%
We omit the $I$ term in each of these definitions if it is clear from
context which  agent computed each of the shares in $S$.

BCG's implementation of VSS and CC guarantees that if $s_i$ is
the share of an honest agent $i$ after running an invocation of
VSS or CC, then there exists a polynomial $p$ of degree $t$
(where $t$ is a bound on the number of malicious agents) such that $p(i) = s_i$
and $p(0)$ is the secret shared through VSS or computed through
CC. Moreover, this polynomial $p$ is uniformly sampled from 
the set of all polynomials $p'$ of degree $t$ with $p'(0) = s$.  With
BCG's implementation of VSS and CC, a pair $(I,S)$ is 
realizable iff there exists a polynomial $p$ of degree $t$
such that $p(i_k) = s_k$ for all $k$. 

With this notation, we can state the properties that we need for our
VSS and CC implementation.
A circuit that computes values in $\mathbb{F}_p$ can be securely
computed by $n$ agents 
if the inputs are shared using VSS, and the addition and
multiplication gates are computed using CC. At the end of the
computation, each agent has a share of the output. The first
property that we require, to simplify our proof, is that for all sets
$T$ of size at most $t$, the output of such a circuit is uniformly
distributed over $\mathbf{F}_p^{|T|}$: 
\begin{definition}
An implementation of VSS 
and
 CC is \emph{$t$-uniform} if, for
\commentout{
  all pairs of protocols $P_1$ and $P_2$ consisting only of VSS
  instances and a CC instance with a single output, and all sets $I$
  of honest agents with $|I| \le t$, the output 
  of $I$ in $P_1$ is identically
distributed to the output of $I$ in $P_2$. 
}
all circuits $C$ 
with a single output
gate, 
 and all sets $I$ of honest agents with
$|I| \le t$, the output
 of $I$ after securely computing $C$   
is uniformly distributed over $\mathbf{F}_p^{|T|}$.
\end{definition}

We actually seem to need a somewhat stronger property than $t$-uniformity:
\emph{conditional $t$-uniformity} (i.e., $t$-uniformity conditional on
the outcome of earlier CC instances).  In general, a $t$-uniform
implementation may not satisfy conditional $t$-uniformity.  For
example, two empty circuits that take the shares of a
single VSS instance as inputs produce identical outputs, which are the
the shares of the VSS.  
Fortunately, it is easy to convert a circuit $C$
to a circuit $C'$ that computes the same secret, and also
satisfies conditional $t$-uniformity conditional on all other CC
instances. Suppose for simplicity that $C$ has a single output
gate. We construct 
$C'$ by having each agent $i$ invoke a $t$-uniform VSS with 0 as
the secret.  Agent $i$ then computes (using CC) the product of the
secrets whose shares it receives, and adds it share of the product to
the output of $C$.  Clearly the players will get the same value with
$C$ and $C'$ after they share their shares, no matter what the
malicious players do. Also, since $C'$ takes as inputs instances of
VSS that are not used in any other circuit, the output of a subset
$I$ with $|I| \le t$ conditional on the output of all other CC
instances is uniformly distributed over $\mathbb{F}_p^{|T|}$. Thus, if
the implementation of VSS and CC is $t$-uniform, we can assume without
loss of generality that we are working with conditionally $t$-uniform
circuits. 

We also require that the shares of a set of at least $t+1$ honest
agents uniquely determine the secret.

\begin{definition}
An implementation $P$ of VSS 
and CC 
is \emph{$t$-determinate} if
\begin{itemize}
  \item[(a)] $P$ is $t$-resilient;
\item[(b)] $P$ is $t$-uniform if there are at most $t$ malicious
  agents; and
\item[(c)] for all pairs $(I,S)$ with $|I| \ge t+1$, if $(I,S)$ is
  realizable, then there exists a unique full extension $(I', S')$ of $(I,S)$. 
\end{itemize}
\end{definition}

BCG's implementation of VSS is $t$-determinate: it is easy to check
that it satisfies clauses (b) and (c) of the definition; BCG prove
that it is  $t$-resilient.
For all of our constructions, we assume that the secret-sharing
scheme used is $t$-determinate.

\subsection{Constructing $\pi'$}\label{sec:implementation}

\commentout{
In this section, we outline the construction of $\pi'$.
The idea is very similar to BCG's
construction of $\vec{\pi}^f$, so we briefly sketch that construction first.
The idea is straightforward.  Each agent $i$ uses VSS to share its
input $x_i$.  The agents then 
agree using consensus (more precisely, using ACS, which
satisfies some additional properties) on which messages have been
correctly shared and 
use circuit computation to compute
(their share of) $f(x_1,\ldots, x_n)$.  Finally, they 
broadcast
the output of the circuit computation, after which each agent can compute
 $f(x_1, \ldots, x_n)$.
}

Our construction of $\vec{\pi}'$ is similar in spirit to
BCG's constrution of $\vec{\pi}_f$.
As we said earlier,
what makes our setting more 
complicated is that the agents send multiple messages to the
mediator, and the mediator sends multiple messages back.
We will need to keep track of which messages are being sent in
response to which other messages.
Moreover, to
get $t$-bisimulation, we need to be able to simulate all possible
behaviors of the scheduler, both with $\vec{\pi}+\pi_d$ and with $\vec{\pi}'$.

For ease of exposition, we begin by giving a naive
construction of $\vec{\pi}'$, which, as we later show, 
does not quite satisfy all the desired properties. However, it
gives the intuition for the actual construction (which
requires only a small modification of the naive construction).
We now sketch the naive construction, then give a detailed description,
and then explain the minor modifications needed to correct the
problems in the naive construction.
\commentout{
For ease of exposition, we present the construction of $\vec{\pi}'$ in
 steps. We begin by giving a naive construction of
$\vec{\pi}'$, which, as we later show, does not quite satisfy all the
desired properties. However, it gives the intuition for the actual
We then show that a small modification of this construction gives a
construction that satisfies parts (a) and (b) of
Theorem~\ref{thm:main}.
}
This
construction may not
satisfy the bound we claimed on the expected number of messages
when the mediator is responsive.
We show in Section~\ref{sec:message-bound}
how to modify the construction so as to satisfy that bound.

When running $\pi'_i$, each agent $i$ simulates its
counterpart running $\pi_i$ except that, rather than sending and
receiving messages from the mediator, $i$  shares messages 
it shares and reconstructs messages using VSS.  
In addition, all
agents use CC to compute the mediator's local history given the
messages shared by the agents and to compute the messages the
mediator sends to the agents according to $\pi_d$, given its local
history.
Note that, after running CC, each agent has a share of the
mediator's message. If this is a message sent by the mediator to
agent $i$, then each agent sends its share to $i$, so that $i$ can
reconstruct the message.

%
To do the simulation, each agent $i$ computes two sequences,
$\{h_{i,k}\}_{k \in \mathbb{N}}$ and $\{h_{d,k}^i\}_{k \in
  \mathbb{N}}$. Each element $h_{i,k}$ in the first sequence represents the
$i$'s local history the $k$th time that $i$
    is scheduled in the simulated interaction with the mediator, while
    each term 
    $h_{d,k}^i$ of 
  the second sequence represents $i$'s share of the mediator's local
  history  
the $k$th time that the mediator is scheduled in  the simulated interaction. 
Of course, these histories depend (in part) on how $i$ does the
simulation.
  In our naive protocol, we assume that all agents get
scheduled in the simulation at times corresponding to when they get
scheduled in the computation of $\vec{\pi}'$, after getting
corresponding messages.   
That is, whenever $i$ is scheduled in 
 $\vec{\pi}'$, it checks all the messages received from the simulated
mediator since the last time it was scheduled in $\vec{\pi}'$, then
simulates itself 
being scheduled in $\vec{\pi}+\pi_d$ after receiving exactly the same
messages in the same order as it did in $\vec{\pi}'$. 
 This means that 
$h_{i,k}$ 
  is constructed by appending to $h_{i,k-1}$ all
  messages received by $i$ and the results of all local computations
  of $i$ between
the  
   $(k-1)$st and $k$th
time that $i$ is scheduled in $\vec{\pi}'$. 
Therefore, 
in the
naive construction, $h_{i,k}$, which is $i$'s view the $k$th time
that $i$ is scheduled in the simulation of the computation of
$\vec{\pi}+\pi_d$, is also part of $i$'s view of the simulation the $k$th
time that $i$ is 
scheduled in $\vec{\pi}'$.
That is, if $i$ has been scheduled $k$ times in $\vec{\pi}'$, then it
is also scheduled exactly $k$ times in the simulation.
As we show later by example, this property prevents us from being
able to simulate all schedulers in the interaction with the mediator,
and is precisely 
why the naive construction does not quite work.  That said, for now we
continue to explain the naive construction.

Note that, in $\vec{\pi}'$,  $i$ 
does not receive the mediator's actual messages in its simulation; rather, it 
receives shares of those messages. 
Agent $i$ 
appends a message to 
$h_{i,k-1}$
only at the point that the message can be reconstructed from the shares of the
message that 
$i$ receives from the other agents.
After computing $h_{i,k}$, $i$ computes
which messages it sends according to $\pi_i$ (given the history it has
simulated) and, for each such message
$\mu$, $i$ shares $\mu$ using VSS.

Computing $\{h_{d,k}^i\}_{k \in \mathbb{N}}$ is more subtle.
We must ensure that all agents agree on what messages should be
appended to $h_{d,k}$ to get $h_{d,k+1}$; otherwise, agents will not
have a consistent view of the mediator's history.    Since,
at any point in the execution of $\vec{\pi}'$,
different agents may have terminated different invocations of
VSS, this requires a little care.
Let $h_{d,0}^i, h_{d,1}^i, h_{d,2}^i, \ldots, h_{d,k}^i$ be the
sequence of shares of the mediator's local history in the simulation
computed thus far by agent $i$. 
We will ensure that, for all
   $k$,  $(h_{d,k}^1, h_{d,k}^2, \ldots,    
      h_{d,k}^n)$ are shares of some local history $h_{d,k}$ of the
mediator in the computation of $\pi$ being simulated, where
      $h_{d,0}$ is the empty sequence $\langle \, \rangle$ and
$h_{d,k}$ is a prefix of $h_{d,k+1}$ (so that the mediator's history
get increasingly longer).
After computing their shares of $h_{d,k}$, agents can
perform a circuit computation to compute the messages the
mediator sends to the agents 
given local history $h_{d,k}$.

We now describe the naive construction of $\vec{\pi}'$ in more detail.
As we said, because our naive construction assumes that $i$ is
scheduled the same number of times in the simulation of $\pi_i$ as in the actual
computation of $\pi_i'$, the $k$th time $i$ is scheduled when running
$\pi_i'$, $i$'s history includes 
simulated histories $h_{i,0}, \ldots, h_{i,k-1}$ and shares of
simulated histories 
$h_{d,0}^i, \ldots, h_{d,k'}^i$ (note that $k'$ might not be equal to $k$).
These simulated histories are the output of local computations, and
thus are recorded in the $i$'s history.
In addition, $i$'s history keeps track of the status of all the invocations of
protocols like VSS and CC in which $i$ participates (including results
of random coin tosses, which we also view as the outcome of computation).
Note that there might be several invocations of the same protocol that
an agent is involved in at the same time; for example, an agent might
invoke VSS several times before any of them complete.
To remove ambiguity, we assume that all invocations of a
protocol are labeled; for example, the first invocation of VSS invoked by
agent $i$ could be labeled (VSS, $i$, $1$), the second one could
  be labeled (VSS,   $i$, $2$)), and so on. 
These labels are communicated to the scheduler using the scheme
presented in Section~\ref{sec:adversaries}. Thus, we can assume
without loss of generality that the scheduler knows the labels.

If agent $i$ is scheduled when it is in such a state, it first
processes all messages received since the last time it was scheduled.
(We assume that all messages received since the last time that $i$ was
scheduled are held in some buffer.)  ``Processing a message'' $\mu$
consists of $i$ playing its part in the protocol to which $\mu$
belongs (which we assume is indicated in  the label of $\mu$); 
if $\mu$ is a share of a simulated mediator message, $i$ checks if it
can reconstruct a new mediator message and, if so, updates $h_{i,k}$
accordingly. After processing all of its new messages, $i$ will have
constructed $h_{i,k+1}$.  Agent
$i$ checks what action(s)
$\pi_i$ takes given input
$h_{i,k+1}$.
If $\pi_i$ outputs a value $v$, then so does $\pi_i'$; if
these actions include sending one or more messages to the mediator,
then $i$ 
shares those messages using VSS instead. 
%
Finally, if possible, $i$ computes its share of the simulated mediator's
local history $h_{d,k'+1}^i$ and computes (along
with the other agents) which messages the mediator sends to the
agents.
We now explain how this is done.


Agent $i$ computes $h_{d,k}^i$ inductively.
Clearly, $h_{d,0}$, the mediator's initially local history, 
is empty (and all agents know this).
To compute $h_{d,0}^i$, 
$i$ simulates a computation of VSS initiated by agent 1
(there is nothing special about agent 1 here; any other agent would do)
with input the empty sequence under the assumption that all
agents are honest, and takes $h_{d,0}^i$ to be $i$'s share of the
output of the computation.
Since VSS is a randomized protocol, to assure consistency,
all agents must use the same random bits in this
computation of VSS; 
we can assume that these random bits
are hardcoded into $\vec{\pi}'$. Note that this is equivalent to just
hardcoding the values of $h_{d,0}^i$ in $\vec{\pi}'$, but viewing $h_{d,0}^i$
as the output of an invocation of VSS will be useful in the future. 
Assuming that $h_{d,k}^i$ has been computed, we show how to compute
$h_{d,k+1}^i$.  The idea is that the agents perform a circuit
computation with inputs $h_{d,k}^i$ and all the new messages
to be appended to 
$h_{d,k}$
(note that each agent has a share of each
of these inputs).  It is critical when running a CC invocation that the inputs
of each honest agent are consistent with the inputs of all
other honest agents participating in the same invocation. More
precisely, for all pairs of agents $i$ and $j$, if $i$'s $\ell$th
input is $i$'s share of the message being shared in some VSS
invocation (VSS, $i'$, $j'$), then 
$j$'s $\ell$th input must be $j$'s share of the same message.
It is not straightforward to ensure this, since $i$ and
$j$ might have 
completed different invocations of VSS at the time that they update
$h_{d,k}^i$ and 
$h_{d,k}^j$ 
respectively. 

Since this issue arises in a number of contexts, we formalize this
notion of consistency. 
Suppose that $\vec{\rho}$ is a joint protocol and $\vec{h}$ is a
history of $\vec{\rho}$. Let $v$ be an invocation of CC 
in $\vec{h}$ in 
which some honest agent $i$ has participated. Invocation $v$ is
\emph{well-defined} if the following holds:  
\begin{itemize}
\item [(a)] All honest agents eventually participate in $v$.
\item [(b)] Suppose that $\alpha$ has $m$ inputs. For each $\ell$ with
  $1 \le \ell \le m$, there exists an invocation $v_l$ of
  VSS or CC that occurred earlier in the computation such that
  each honest agent $i$'s share of the $\ell$th 
  input of $v$ is $i$'s share of the output of  $v_\ell$.
\end{itemize}
All the invocations of CC to compute $h^i_{d,k+1}$ are
well-defined in this sense, since players must use the shares of the
same secret at each gate of the CC. 
\commentout{
Following the arguments given above, it is critical that all the invocations of CC used to compute $h^i_{d,k+1}$ are
well-defined in this sense.
}
To ensure this, agents first agree on which subset of messages should
be appended to $h_{d,k}$, then they agree on the order in which these
messages should be appended, and finally they append these messages to
$h_{d,k}$ and compute the messages sent by the scheduler to the
agents, which are also appended to $h_{d,k}$.
%
The protocol for extending $h^i_{d,k}$ to $h^i_{d,k+1}$ proceeds in
four phases, denoted $k1, \ldots, k4$.
\begin{enumerate}
  \item[Phase $k1$:] 
Let $N$ the maximum number of messages that an honest agent sends when
running $\vec{\pi}+\pi_d$.
Each agent $i$ participates in 
$nN$ consensus  
protocols in phase $k1$, denoted $p_{1,1,k}, \ldots, p_{n,N,k}$,
where $p_{i,j,k}$ is intended to achieve consensus on whether $i$ has
shared its $j$th message
successfully.  
More precisely, $i$'s
  input to consensus protocol $p_{j,\ell,k}$ is 1 iff $i$ has terminated
    $j$'s $\ell$th   invocation of VSS 
    by the time $i$ starts
     phase $k1$
     and $p_{j,\ell,k'}$ has output 0 for all $k' < k$.
        Agent $i$ waits until it has terminated all the $p_{j,\ell,k}$
        consensus protocols it is involved with in phase $k1$ 
     before starting phase $k2$. If the output of some consensus
     protocol      $p_{j,\ell,k}$ is   1,
     then $i$ waits until it has also completed $j$'s $\ell$th VSS invocation
     in round $k$ before starting  phase $k2$.
    \item [Phase $k2$:] 
Let $p_{j_1, \ell_1,k}, p_{j_2, \ell_2,k}, \ldots, p_{j_{m_k},
    \ell_{m_k},k}$ be the consensus protocols that were used in phase
$k1$ and had output 
1, ordered in lexicographic order (i.e., $p_{j,\ell,k}$ precedes
$p_{j', \ell', k'}$ iff $j < j'$ or [$j = j'$ and $\ell < \ell'$]).
In this phase, $i$ coordinates with the other agents on the order
that they 
should append
the messages shared in 
(VSS, $j_1, \ell_1$), \ldots, (VSS,
    $j_{m_k}, \ell_{m_k}$) to $h_{d,k}$. 
We want the agents to agree on the same permutation of $(j_1,\ell_1),
\ldots, (j_{m_k},\ell_{m_k})$.  
To do this, they use BCG's secure computation protocol.  Each agent
$i$ inputs a permutation to the protocol; the output is the
permutation that they coordinate on.
Agent $i$'s
input to the protocol is the 
 unique permutation $\sigma_i: [m_k] \longrightarrow [m_k]$ satisfying
  $\sigma_i(a) < \sigma_i(b)$ iff $i$ terminated (VSS, $j_a, \ell_a)$
 before (VSS, $j_b, \ell_b)$. Note that since $i$ completes all the 
VSS invocations that are in progress in phase $k1$ before it starts
phase $k2$, $i$ can compute this permutation.
The BCG computation returns the (unique) permutation
$\theta$ satisfying $\theta(a) < \theta(b)$ iff 
$\sum_{i \le n}
    \sigma_i(a) <  \sum_{i \le n} \sigma_i(b)$ or $\sum_{i \le n}
    \sigma_i(a) =  \sum_{i \le n} \sigma_i(b)$ and $a < b$.
Thus, roughly speaking, 
$\theta(a) < \theta(b)$ if, on average, agents 
    terminated $(VSS, j_a, \ell_a)$ before $(VSS, j_b, \ell_b)$.
    Note that because of the asynchrony of the system, we can 
    guarantee only that at most $n-t$ inputs will be available when
    computing $\theta$. For the remaining inputs $\sigma_i$ we 
take $\sigma_i$ to be the identity permutation.
\item[Phase $k3$:]  
Agent $i$ uses CC to append new messages to $h_{d,k}^i$. More precisely,
it updates $h_{d,k}^i$ 
with (VSS, $j_{\theta(1)}, \ell_{\theta(1)}$), ..., (VSS, $j_{\theta(m_k)}, \ell_{\theta(m_k)}$)
in the order determined by the permutation $\theta$ computed in phase $k2$.
Note that the properties of VSS 
and the fact that $i$ completes all outstanding VSS invocations in phase $k1$ 
guarantee that
all honest agents have a share of all these messages
when they start phase $k3$.
This procedure gives agent 
$i$ 
a share that we denote $\tilde{h}^i_{d,k+1}$ of the mediator's updated
local history  $\tilde{h}_{d,k+1}$
after appending these new messages to $h_{d,k}$ in the appropriate order.
\item [Phase $k4$:] Agent $i$ computes $h_{d,k+1}^i$ by using CC to
  append to $\tilde{h}_{d,k+1}$ the message the mediator sends
    to the according to $\pi_d$, given input 
    $\tilde{h}_{d,k+1}$.
    (This can be done because all agents know the mediator's protocol
        $\pi_d$.)
    Agent $i$'s input for this invocation of CC is
  its share $\tilde{h}^i_{d,k+1}$.
    Note that each agent invokes CC only once, using it to compute all
    the mediator's messages are computed and appended them to
  $\tilde{h}_{d,k+1}$. Agent $i$'s output of this invocation of CC is its share
  of $h_{d,k+1}^i$. 
\end{enumerate}

In phase $k4$ agents never actually compute (their
shares, if any, of) the messages sent by the mediator during its $(k+1)$st turn;
they compute only the result of appending these
message to $h_{d,k}$.  Later we will 
see how agents compute their shares of each of thesee messages
individually, using the fact that it is encoded in $h_{d,k+1}$. 

This protocol satisfies
two 
important properties
if $n > 4t$:


\begin{lemma}\label{lemma:well-defined}
    All honest agents eventually terminate phase $k1$.
    Moreover, for all adversaries of size at most $t$ and all histories, the CC
protocol invoked in phase $k3$ is well-defined. 
\end{lemma}

\begin{proof}
If the output of some consensus protocol 
$p_{j,\ell,k}$
is 1, then the
properties of VSS guarantee that at least one honest agent had
    input 1. Thus, 
at least one honest agent terminated 
(VSS, $j$, $\ell$). 
The
properties of VSS guarantee that all other honest agents 
eventually terminate this VSS invocation as well.
The properties of consensus and secure computation guarantee that all
agents use the outputs of the same VSS invocations in the same order,
which means that the CC procedure of phase $k3$ is well-defined for
all runs. 
\end{proof}

\begin{lemma}\label{lemma:correctness}
If an honest agent shares a message $\mu$ using VSS, then $\mu$ will
be in $h_{d,k}$ for some $k$, and hence each honest agent $i$ will have
a share of $\mu$ in $h_{d,k}^i$.
\end{lemma}

\begin{proof}
\commentout{
 If (VSS, $j$, $k$) is invoked by an honest agent $j$, then all honest
agents are guaranteed to eventually terminate this invocation of
VSS.  After all the honest agents 
terminated 
(VSS, $j$, $k$),
the output of 
the corresponding consensus protocol
$p_{j,k}$ (at Step 1) 
is guaranteed to be 1 since all honest agents have
input 1,
which guarantees that (VSS, $j$, $k$) is included in the CC
procedure at Step 2. 
}
 If (VSS, $j$, $\ell$) is invoked by an honest agent $j$, then all honest
agents are guaranteed to eventually terminate this invocation of
VSS. Thus, the output of consensus protocol $p_{j,\ell, k}$ (at Phase
$k1$) is 1 for exactly one value of $k$. (Note that if the output of
$p_{k, \ell, k}$ is $1$, then all honest players take 0 to be the input for
all consensus protocols $p_{j, \ell, k'}$ with $k' > k$.  This
guarantees that the output of all these protocols is 0.) This ensures
that (VSS, $j$, $\ell$) is appended to $\tilde{h}_{d,k+1}$ in Phase
$k3$, and thus it is included in $h_{d, k+1}$. 
\end{proof}

Since, by the time each agent $i$ finishes computing $h_{d,k}^i$,
all
the messages that the mediator sends to each agent are already
encoded in $h_{d,k}$,
it may seem that to compute the shares of
these messages individually,
$i$ would have to use an instance of CC
for each one.

However, this procedure is not so straightforward since $i$ does not
know beforehand how many messages the mediator sends or the order in
which the mediator sends messages the $k$th time it is scheduled
(although this is also encoded in $h_{d,k}$). To deal with this 
issue, before computing its share of each of the mediator's messages, $i$ first checks if there is a message that still needs to be sent 
and, if so, who the recipient is.

More precisely, let $f_{k, \ell}$ be the function that takes as input
a mediator's local history and returns the recipient of the $\ell$th
message sent by the mediator the $k$th time it is scheduled; similarly, let
$g_{k, \ell}$ be the 
function that computes the $\ell$th message sent by the mediator
the $k$th time it is scheduled, given the mediator's history. If the
mediator sends fewer than $\ell$ messages the $k$th time it is
scheduled, or if the input is not a well-defined 
local history, both $f_{k,\ell}$ and $g_{k,\ell}$ return 0. After
computing $h_{d,k}$, agent 
$i$ proceeds as follows for $\ell = 1,2, \ldots$: it performs a
circuit computation of $f_{k, \ell}$ 
with input $h_{d,k}^i$. Then $i$ broadcasts the output of this
computation, and uses the values it receives from other agents to
reconstruct $f_{k,\ell}(h_{d,k})$.
\commentout{
If $f_{k, \ell}(h_{d,k}) = 0$, $i$
computes $h_{d,k+1}$ and terminates.  If 
$f_{k, \ell}(h_{d,k}) = j \not= 0$, $i$ performs a circuit computation of 
$g_{k, \ell}$ with input $h_{d,k}^i$ and sends its output to $j$.
}
If $f_{k, \ell}(h_{d,k}) \not= 0$, $i$ performs a circuit
computation of $g_{k, \ell}$ with input $h_{d,k}^i$ and
computes $f_{k, \ell + 1}  (h_{d,k})$.
 If $f_{k, \ell}(h_{d,k}) = 0$, then for each $\ell' < \ell$, $i$
sends the output of its circuit computation of
$g_{k, \ell'}$ to agent 
$f_{k, \ell'}(h_{d,k})$.
\commentout{
After sending these messages, $i$ computes
$h_{d,k+1}$ and terminates. 
}

\commentout{
After agent $i$ computes $h_{d,k}^i$, $i$ participates in $Nn$ CCs for
phase $k$, where CC $(k,j,m)$ results in $i$ having a share of the
$m$th message that the mediator sends $j$ the $k$th time that the
mediator is scheduled (where we take the $m$th message to be $\bot$ if
the mediator sends $j$ fewer than $m$ messages).  Then $i$ sends its
share of $j$'s message to $j$ (so that $j$ can reconstruct the
mediator's messages). 
}

This completes the description of the naive version of $\vec{\pi}'$.
As we have been hinting, this protocol does not quite work.  The
following example makes the reasons more precise.

Consider 
a protocol $\vec{\pi} + \pi_d$ 
 in which the mediator sends 
a STOP message to
each agent the first 
time it (the mediator) is scheduled.
If $i$ was scheduled before receiving the STOP
message, it outputs 0; otherwise, it outputs 1. 
Note that any
combination of outputs is possible with $\vec{\pi}+\pi_d$, depending
on when the scheduler 
schedules the mediator and the agents.
However, this is not true 
for $\vec{\pi}'$
as we have defined it.
Suppose, for example, that all agents are honest, and $i$ is the first
agent scheduled in a history of $\vec{\pi}'$.  At this point, $i$ is
supposed to compute $h_{i,1}$.  Since it has not received any
messages, it will take $h_{i,1}$ to be empty, and thus output 0.  It
follows that no history of $\vec{\pi}'$ can end with all agents
outputting 1, which means that $\vec{\pi}'$ does not
$t$-bisimulate $\vec{\pi} + \pi_d$. 
In our construction  of the naive version of $\vec{\pi}'$, each agent
$i$ calculates $h_{k,i}$ 
 the $k$th time that $i$ is scheduled
in $\vec{\pi}'$. 
However, since computing each of $h_{d,1}, h_{d,2}, \ldots $ takes
several turns of $i$, the mediator's history $h_{d,k'}$ being computed
by $i$ during its $k$th turn satisfies that $k' \le k$. This means
that $i$ is simulating that the mediator, at all times, has taken less
turns than $i$, which may not be true in the protocol with the
mediator. As our example shows, some scenarios cannot be simulated with
our naive construction because of this. 
\commentout{
But, as our example shows, there may be some
histories of $\vec{\pi} + \pi_d$ that cannot be simulated this way,
because not all agents will have the information they need to do this
calculation.
}
%
%
We deal with this problem by using the scheduler in the 
simulation to
determine whether an update to $h_{i,k}$ or $h_{d,k'}^i$ should occur when
$i$ is scheduled.

We proceed as follows.
When an agent $i$ is first
scheduled, $i$ sends two special messages, $\proceed_i$ and $\proceed_{d,0}$,
to itself and 
 computes $h_{i,0}$ (which is just an empty history) and $h_{d,0}^i$.
What $i$ does when it is scheduled for the $\ell$th time for $\ell >1$
depends on whether it
has received messages of the form $\proceed_i$ and $\proceed_{d,r}$ and
messages from itself since the last time it was scheduled. 
Suppose that $i$ has computed the sequences $h_{i,0}, \ldots, h_{i,k}$
and $h_{d,0}^i, \ldots, h_{d,k'}^i$ when  it is scheduled for the
$\ell$th time. 
If $i$ has not received a $\proceed_i$ message since the last time it
was scheduled,
$i$ does not  compute $h_{i,k+1}$.
If $i$ has received a $\proceed_i$ message since the last time it was
scheduled, 
then it sends itself another $\proceed_i$ message and 
computes $h_{i,k+1}$ as described above, using all the messages it
received since it was last scheduled and received a $\proceed_i$ message
Thus, $i$ computes the next history in the sequence 
$\{h_{i,k}\}_{k\in \mathbb{N}}$ if and only if $i$ receives a
$\proceed_i$ message.
Similarly, if $i$ has not received a message of the form
$\proceed_{d,r}$ since the 
last time it was scheduled, then it
does not do any of the steps needed to compute $h^i_{d,k'+1}$.  If it
has received a message of the form $\proceed_{d,r}$ message since the
last time it 
was scheduled, it sends itself a message of the form
$\proceed_{d,r+1}$.  (Thus, the second component of the subscript
serves a counter for the number of such messages that have been sent.)
If $k' \le r$, then $i$ plays its part in computing $h_{d, k'+1}^i$.
Otherwise, $i$ does not take part in any procedure involved
   in the computation of 
      $h_{d, k'+1}^i$; that is, $i$ waits until
   it receives  $\proceed_{d,r}$ before attempting to
   compute $h_{d,r}^i$. 
Thus, when it is scheduled, $i$ may take part in 
computing both $h_{i,k+1}$ and
$h_{d,k'+1}^i$, only one of them,  or neither of them.
Since the scheduler must eventually deliver all messages,
all agents receive all the $\proceed$ messages that they send
themselves, so eventually do update $h_{i,k}$
and $h_{d,k}^i$.

This completes the construction of $\vec{\pi}'$.  In the next few
subsections, we prove that $\vec{\pi}'$ has the desired properties.

\subsection{The proof of Theorem~\ref{thm:main}(a)}\label{sec:proof}

\commentout{
In this proof, we assume that the adversary
is a single entity that controls the malicious agents and the scheduler.
As we showed in \cite{game-theory-paper}, this assumption can be
made without loss of generality.  We further assume (again, without
loss of generality) that malicious agents do not communicate directly
among themselves, but make their decisions knowing the histories of
the other malicious agents.
}

We now prove Theorem~\ref{thm:main}(a).
For ease of exposition, we begin by proving this result for the
special case that $t' = t$, showing that $\vec{\pi}'$
$t$-bisimulates $\vec{\pi}$ if $n > 4t$. 

%
\textbf{Proof that $\vec{\pi}'$ $t$-bisimulates $\vec{\pi}$ if $n > 4t$:}
We actually prove a result slightly stronger than
Theorem~\ref{thm:main}(a): while the definition of bisimulation allows  
$\sigma_e'$ to depend on both $\sigma_e$  $\vec{\tau}'_T$ and 
$\vec{\tau}'_T$ to depend on both $\vec{\tau}_T$ on $\sigma_e$, in our
construction below, 
$\sigma_e'$ depends only on $\sigma_e$ (and not on $\vec{\tau}'_T$),
while $\vec{\tau}'_T$ depends only on $\vec{\tau}_T$ (and not on $\sigma_e$). 

We begin by showing that $\vec{\pi}'$ satisfies part (a) of the
definition of bisimulation assuming that all players are
honest. Later, we show how this proof can also be applied to the case
in which a subset $T$ of players deviate. 
Given a scheduler $\sigma_e$ in the mediator setting, we construct a
scheduler $\sigma_e'$ in the setting without the mediator as follows.
Initially, 
$\sigma_e'$ schedules each agent $i$ exactly once.  Recall
that if $i$ is honest, the first time it is scheduled it sends only
$\proceed_i$ and $\proceed_{d,0}$ messages to itself.
The point of scheduling all the agents initially is simply to
get these $\proceed$ messages into the system.
From then,  
just as the agents do with $\vec{\pi}'$, 
$\sigma_e'$ simulates which history $h_e$ the scheduler
would have  
\commentout{
when the
scheduler $\sigma_e'$ moves, it first computes (given its 
local history $h_e'$) what local history $h_e$ it would have
}
in the interaction with the mediator if the mediator and the agents
used $\vec{\pi}+\pi_d$ and the scheduler used $\sigma_e$.
At the beginning of the game, the scheduler sets $h_e$ to the empty
history. How the scheduler updates $h_e$ and what actions the
scheduler performs according to $\sigma_e'$ then depend on the form
of $a = \sigma_e(h_e)$ (i.e., the actions that $\sigma_e$ would
perform given history $h_e$), and on the actions that the players
perform afterwards:   
\commentout{
What the scheduler does according to $\sigma_e'$ then
depends on the form of  $a = \sigma_e(h_e)$,
and how the scheduler 
simulates
 $h_e$ depends on 
$a$ and the actions that the players perform afterwards:
}
\begin{itemize}
\item If $a$ has the form $\Sch(i)$, then it delivers $i$'s
  most recent $\proceed_i$ message
    if there is one to deliver,
  and then schedules $i$.
    Suppose that $i$ initiates $\ell$ VSS instances during its turn. Then
    immediately after $i$'s turn, $\sigma_e'$ appends $\Sch(i)$ and
  $\ell$ $\Snd(d, i)$ events to $h_e$, followed by a $\Done(i)$
  event. 
\item If $a$ has the form $\Sch(d)$ and it is the $k$th time that the
    mediator is scheduled according to $\sigma_e$, then the scheduler
    delivers to each agent $i$ 
    its $\proceed_{d,k}$ message
  if there is one to deliver
    and then schedules agents
    cyclically ($1,2,\ldots,n, 1, 2,\ldots$) until all agents
    $i$ finish  computing $h_{d,k}^i$ and their share of each of the
    messages sent by 
    the mediator 
during its $k$th turn.
    The scheduler also
delivers to each agent $i$ all the messages required by $i$ for the computation of $h_{d,k}^i$ and
    the shares of the mediator's messages immediately after they are
  sent.  
  Suppose that the players determine that the mediator sends messages to
  $j_1, j_2, \ldots, j_\ell$, in that order. Then $\sigma_e'$ appends
  $\Sch(d), \Snd(j_1, d), \ldots, \Snd(j_\ell, d), \Done(d)$ to
  $h_e$.
Note that the scheduler knows $j_1, \ldots,
    j_\ell$, and also knows when each player $i$ terminates the
    computation of $h_{d,k}^i$ (given our assumption that the scheduler
        knows the label of each message), 
 since  player $i$ has terminated the computation of $h_{d,k}^i$ if all
        messages related to 
  this computation have been delivered and no player sent further
  messages
when it was scheduled.
Thus, all players $i$ are guaranteed to have terminated the
computation of $h_{d,k}^i$ after
the scheduler has gone through a full cycle of scheduling the players 
without any player sending any message required for the computation of 
$h_{d,k}^i$, for $i=1, \ldots, n$.
(Recall that we are assuming for now that all players are honest.)
\item If $a$ has the form $\Rec(i,d,\ell)$,
  the scheduler delivers to $i$ all the messages that $i$
    needs to compute
    the mediator's $\ell$th message to $i$. That is, the scheduler
    delivers     the messages from other agents
    containing the shares of the $\ell$th message from the mediator to
    $i$.  (By our inductive hypothesis, these messages have been sent
    but not yet delivered.)
        Then $\sigma_e'$ appends $\Rec(i,d,\ell)$ to $\sigma_e$.
      \item If $a$ has the form $\Rec(d,i,\ell)$, the scheduler 
        schedules the agents cyclically until all the agents finish computing
    $(VSS,i,\ell)$.
        More precisely, the scheduler delivers only the messages involved in
        protocol $(VSS, i, \ell)$, and does so immediately after they
        are sent, 
all of this while scheduling  
the agents cyclically until all the agents stop sending
messages. 
Then $\sigma_e'$ appends $\Rec(d,i,\ell)$ to $\sigma_e$.

    \commentout{
      \item If $a$ has the form $\Rec(d,i,\ell)$, the scheduler 
        schedules the agents cyclically
        until all the agents finish computing $(VSS,i,\ell)$,
and delivers all the messages required for the
agents to compute $(VSS,i,\ell)$ immediately after they are sent.
}
\end{itemize}

Note that $\sigma_e'$ does not depend on the protocol $\vec{\tau}_T$
used by  malicious agents.
Suppose that $\vec{\pi} + \pi_d$ and $\sigma_e$ are deterministic.
For each input $\vec{x}$, let $J_i^{k}(\vec{x})$ denote agent $i$'s
local history at the end of its $k$th turn in the unique history of 
$(\vec{\pi} + \pi_d, \sigma_e, \vec{x})$. When the agents
use $\vec{\pi}'$, they simulate the computation of $\vec{\pi} + \pi_d$.
Let  $K_{i}^{k}(\vec{x})$ denote $i$'s history at the
end of $i$'s $k$th turn in the simulation.  
Although $\vec{\pi}'$
randomizes, since $\vec{\pi} + \pi_d$ and $\sigma_e$ are
deterministic, as we now show, the value of $K_i^k(\vec{x})$ is 
independent of this randomization.

\begin{lemma}\label{lemma:easy-implication}
  For all input profiles $\vec{x}$,
  $K_i^{k}(\vec{x}) = J_i^{k}(\vec{x})$.
\end{lemma}


We prove this lemma by proving a more general result that
establishes a correspondence between histories of $\vec{\pi}$ and
histories of $\vec{\pi}'$. In $\vec{\pi}'$, agents attempt to simulate
all the events of $\vec{\pi}+\pi_d$, which include being
scheduled and sending and receiving messages.  By the construction
of $\sigma_e'$, all shares of a message sent by the mediator
are received by its recipient virtually ``at the same time''
(more precisely, they are received one immediately after the other, 
with no other action in between). This
allows us to define a correspondence between events in a history $h$
of $\vec{\pi}+\pi_d$ when used with scheduler $\sigma_e$ and 
events in a history 
$h'$ of $\vec{\pi}'$ when used with scheduler $\sigma_e'$. 
We start by defining
the correspondence between events that are in an agent $i$'s history
in $h$ and $h'$.
\begin{itemize}
\item The event that agent $i$ is scheduled for the $k$th time in
    $h$ corresponds to the event that $i$ is scheduled after receving
    its $k$th $\proceed_i$ message in $h'$.
  (Of course, $i$ may not have received $k$ $\proceed_i$ messages in
  $h'$; in this case, no event in $h'$ corresponds to the event of $i$
  being scheduled for the $k$th time in $h$.  Similar comments hold
  for all the other correspondences defined 
 below.)
  \item The event that $i$ sends its $\ell$th message in $h$
    corresponds to the event that $i$ initiates its $\ell$th 
        invocation of VSS in $h'$.
  \item The event that $i$ receives the $\ell$th message sent by the mediator
    during its $k$th turn in $h$ corresponds to 
    the event that $i$ receives a share
    of $g_{k,\ell}(h_{d,k})$ in $h'$. Note that
        $g_{k,\ell}$ encodes the $\ell$th message sent by the
    simulated mediator during its $k$th turn given its local history
    $h_{d,k}$. 
\end{itemize}

\commentout{
Events in the mediator's history in $h$ do not correspond to single
events in $h'$.  Rather, they correspond to a collection of events.
The scheduler $\sigma_e'$ guarantees that these events happen
consecutively in $h'$.
}
Since the mediator is being simulated by all agents, events in the
mediator's history in $h$ do not correspond to single 
events in $h'$. Rather, they correspond to exactly $n$ events, one for
each agent. Scheduler $\sigma_e'$ guarantees that these $n$ events
occur
consecutively in $h'$. 
\begin{itemize}
\item The event that the mediator is scheduled for the $k$th time in
  $h$ corresponds to 
    the set of events in $h'$ consisting of agent $i$ being scheduled
    after receiving a  $\proceed_{d,k}$ message, for each agent $i$. 
\commentout{
\item The mediator sends its $\ell$th message during the $k$th time it is scheduled to agent $j$.
\item The mediator receives agent $j$'s $\ell$th message. 
  }
\item The event that the mediator 
sends the
 $\ell$th message 
to agent $j$ during its $k$th turn  in $h$
    corresponds to the set of events in $h'$ consisting of each agent $i$
    computing its share of 
    $g_{k,\ell}(h_{d,k})$ (which encodes the
$\ell$th message sent by the simulated mediator during its $k$th turn)     
    and sending it
    to agent $j$.
  \item The event that the mediator receives $j$'s $\ell$th message in
        $h$ corresponds to the set of events in $h'$ consisting of each agent
    $i$ terminating the $\ell$th VSS invocation initiated by $j$.
  \end{itemize}


\commentout{
Since all agents take part in the mediator's simulation, the events in $\vec{\pi}$ involving the mediator correspond to $n$ events in $\vec{\pi}'$, one for each agent. The events relative to agent $i$ are the following:

\begin{enumerate}
\setcounter{enumi}{3}
\item Agent $i$ is scheduled after receiving a $\proceed_{d,k}$ message.
\item Agent $i$ computes the share of the mediator's $\ell$th message during its $k$th turn and sends it to agent $j$.
\item Agent $i$ terminates $j$'s $\ell$th VSS instance.
\end{enumerate}
}

Note that we have not included the $\Done(i)$ events in 
the correspondence.  Even though such events are needed to define then
end of agent$i$'s turn, they are
redundant, since they come immediately before a $\Sch(i)$ in $i$'s
local history.  

An event in a history $h'$ of $\vec{\pi}'$ is a \emph{simulation event} if it
could correspond to some event in another history $h$ of
$\vec{\pi}$. More precisely, an event $e$ in history $h'$ is a
simulation event if there exists a 
history $h$ of $\vec{\pi}$ and an event $e$ in $h$ such that $e'$
corresponds to $e$.
Two histories $h$ and $h'$ of $\vec{\pi}$ and $\vec{\pi}'$
\emph{correspond} 
if all non-$\Done$ events in $i$'s history in $h$ correspond to some
event in $h'$, all non-$\Done$ events in the mediator's
history in $h$ correspond to $n$ events in $h$, one for each agent,
each simulation event in $h'$ corresponds to some event in $h$, and 
the order of corresponding events in each agent $i$'s history is
preserved; more precisely,   if $e_1$ and $e_2$ are two non-$\Done$
events in $i$'s (resp., the mediator's) history in $h$, and $e_1'$ and
$e_2'$ are the events that correspond to $e_1$ and $e_2$ in $i$'s
history in $h'$, then $e_1$ precedes $e_2$ in $h$ iff $e'_1$ precedes
$e_2'$ in $h'$.  
\commentout{
More precisely, two histories $h$
and $h'$ of $\vec{\pi}$ and $\vec{\pi}'$ correspond if  following
holds:
%
\begin{itemize}
\item[(a)] For each 
non-trivial 
event $e'$ in $h'$ there is a corresponding event $e$ in $h$.
\item[(b)] For each 
non-$\Done$ 
event $e$ in $i$'s local history $h_i$, there exists a corresponding event $e'$ in $h'_i$. Moreover, if an event $e_1$ comes before another event $e_2$ in $h_i$, their corresponding events $e'_1$ and $e'_2$ satisfy that $e'_1$ comes before $e'_2$ in $h'_i$.
\item[(c)] For each 
non-$\Done$ 
event $e$ in the mediator's local history $h_d$,
for each agent $i$ there exists a corresponding event $e'_i$ in
  $h'_i$. Moreover, if an event $e_1$ comes before another event $e_2$
  in $h_d$, their corresponding events 
  in $i$'s local history $e'_{i,1}$ and $e'_{i,2}$ satisfy that $e'_{i,1}$ comes before $e'_{i,2}$ in $h'_i$.
\end{itemize}
}

Lemma~\ref{lemma:easy-implication} follows from the following lemma,
which is almost immediate from the construction of $\sigma_e'$ and
$\pi'$.  Although there are a number of histories in 
$\vec{\pi}'$ with scheduler $\sigma_e'$ and input $\vec{x}$
 due to the randomization used in protocols such as VSS and
CC, all of them correspond to the unique history in 
$\vec{\pi}$ when the scheduler plays $\sigma_e$ and players have input
profile $\vec{x}$.

\begin{lemma}\label{lem:J=K}
\commentout{
 for each agent $i$
\begin{itemize}
\item[(a)] The sub-sequence of events in $i$'s history in $\vec{\pi}'$
  involving agent $i$ (of types 1,2 and 3) are the corresponding
  events of $i$'s history in $\vec{\pi}$. 
\item[(b)]  The sub-sequence of events in $i$'s history in $\vec{\pi}'$ involving the mediator (of types 4,5 and 6) are the corresponding events of the mediator's history in $\vec{\pi}$.
\end{itemize}
}
For all input profiles $\vec{x}$, the unique history $h$ of 
where the agents use $\vec{\pi}$ with scheduler $\sigma_e$
and input profile $\vec{x}$  
corresponds to all the histories $h'$ of
where the agents use $\vec{\pi}'$ with scheduler $\sigma_e'$ and
input profile $\vec{x}$.
\end{lemma}

Note that Lemma~\ref{lem:J=K} implies Lemma~\ref{lemma:easy-implication},
since it states that agents simulate receiving and sending messages in
exactly the same order with $\vec{\pi}'$ as they do with $\vec{\pi}+\pi_d$.
Moreover, if the protocol is deterministic, the contents
of those messages are uniquely determined. 

If $\vec{\pi} + \pi_d$ or $\sigma_e$ involve randomization, we can
assume that the agents and mediator toss all the coins they need 
at the beginning (before they are first scheduled)
and then use the outcomes of these coin tosses for their
decisions. Fixing the outcome of such coin tosses makes the protocols
deterministic, and an analogous argument to that used for
Lemma~\ref{lemma:easy-implication} for each of the possible
sequences of coin
tosses guarantees that the agents' outputs are identically distributed
in $\vec{\pi} + \pi_d$ and in $\vec{\pi}'$. Since 
 $(\vec{\pi}_{-T},\vec{\tau}_T)$ is just another protocol,
it immediately follows from Lemma~\ref{lemma:easy-implication} that
for all input profiles 
\commentout{
$\vec{x}$, $O(\vec{\pi}_{-T}, \vec{\tau}_T, \sigma_e,
\vec{x})$ and $O(\vec{\pi}'_{-T}, \vec{\tau}'_T,
\sigma_e', \vec{x})$
}
$\vec{\pi}(\vec{x}, (T, \vec{\tau}_T, \sigma_e))$ and $\vec{\pi}'(\vec{x},(T, \vec{\tau}'_T, \sigma'_e))$
 are identically distributed for all
possible protocols $\vec{\tau}_T$ for the malicious agents.
This completes the proof that part (a) of the definition of
``bisimulates'' holds.

\commentout{
\begin{proof}[Proof of Lemma~\ref{lemma:easy-implication}]
We show the lemma for the case in which $\vec{\pi}$ and $\sigma_e$ are
deterministic, although the same argument works if the
strategies include randomization. 

[TODO]: We show by induction on $n$ that all prefixes of histories of up to $n$ blocks of events are equally distributed in both scenarios.
\end{proof}
}


We now prove that part (b) of the definition of ``bisimulates'' holds.
For this proof, we assume without loss of generality that
malicious agents output their local history when running
$\vec{\tau}'_T$, since any output must be
a function of their local history.
%
For ease of exposition, we begin by giving the highlights of the
construction of $\vec{\tau}$ and $\sigma_e$, given $\vec{\pi}'$,
$\vec{\tau}'_T$, 
 and $\sigma_e'$; we later present the
construction in more detail. The idea for constructing
$\vec{\tau}_T$ 
 and $\sigma_e$ is that the adversary simulates
what would occur if honest agents use $\vec{\pi}'_{-T}$,
malicious agents use $\vec{\tau}'_T$, and the scheduler uses
$\sigma_e'$. If the adversary
in the protocol with the mediator
knew the input $\vec{x}_{-T}$ of honest
agents, 
\commentout{
this would be straightforward: it could perform the 
simulation before the game starts, and send/deliver messages in such a
way that they match what honest agents compute in their local
histories $h_{i,k}$ and $h_{d,k}$. 
}
the adversary could perform the simulation before the protocol starts, and
have malicious agents output the local history they have in the
simulation, regardless of their history in $\vec{\pi}$. 
However, the adversary does not
know the honest agents' inputs.
Thus, the adversary does the simulation assuming honest agents
have some fixed input, which we take to be $\vec{0}_{-T}$.
The following lemma makes precise the sense in which using
$\vec{0}_{-T}$ rather than $\vec{x}_{-T}$ is ``safe''.

\commentout{
\begin{definition}
  Denote by $J_T^e(\vec{\sigma}, \sigma_e, \vec{x})$ the random variable
whose values are the local history profiles of the agents in $T$ and the
environment when the agents 
 use
$\vec{\sigma}$, the scheduler plays $\sigma_e$,  and agents have input
$\vec{x}$. We say that a protocol $\vec{\sigma}$ is $t$-uniform if for
all subsets $T \subseteq [n]$ with $|T| \le t$, all strategies
$\vec{\tau}_T$ for agents in $T$, all input profiles $\vec{x}_T$ for agents
in $T$, all schedulers $\sigma_e$,  and all pairs of 
input profiles $\vec{x}_{-T}, \vec{x}_{-T}'$ for agents not in $T$, we have that
$J_T^e(\vec{\sigma}_{-T}, \vec{\tau}_T, \sigma_e, \vec{x}_{-T},
\vec{x}_T)$ and $J_T^e(\vec{\sigma}_{-T}, \vec{\tau}_T, \sigma_e,
\vec{x}'_{-T}, \vec{x}_T)$ are identically distributed.
\end{definition}
}

\begin{lemma}\label{lemma:secret}
Let $J(\vec{\pi},A,\vec{x})$ be a random a random variable whose
values are the values of the malicious agents, when the honest agents
use $\vec{\pi}$, given input profile $\vec{x}$, and adversary
$A = (T,\tau_{T}, \sigma_e)$.
Let $\vec{\sigma}_{VSS}$
and $\vec{\sigma}_{CC}$ be the implementation of VSS and CC,
respectively, in a $t$-resilient secret-sharing scheme.  
Then for all adversaries $A = (T,\vec{\tau}_T,\sigma_e)$ and input
profiles $\vec{x}$ and $\vec{x}'$, we have
\begin{itemize}
  \item $J^e((\vec{\sigma}_{VSS}), A,   \vec{x},)$ and
  $J^e((\vec{\sigma}_{VSS}), A, 
(\vec{x}_{-T}, \vec{x}'_{-T}))$ are identically distributed; 
  \item $J^e((\vec{\sigma}_{CC}), A,   \vec{x},)$ and
  $J^e((\vec{\sigma}_{CC}), A, 
(\vec{x}_{-T}, \vec{x}'_{-T}))$ are identically distributed; 
\end{itemize}
\end{lemma}

Lemma~\ref{lemma:secret} implies that the adversary's history in an invocation of
VSS and CC is independent of the actual inputs of the honest
agents. This follows easily from the definition of $t$-resilience,
since otherwise the adversary could deduce information about the honest
agents' inputs given its local history. This means that much
of the simulation can be performed by the adversary without having to
know which inputs honest agents are using.  
Indeed, 
there are only three types of actions or decisions of an honest
agent $i$ that both depend on the honest agents' inputs and can affect the
adversary's local history: 
\begin{itemize}
  \item [(a)] how many time $i$ invokes VSS 
  after receiving a $\proceed_i$ message; 
  \item [(b)] what values $i$ broadcasts after computing its
       share of $f_{k,\ell}(h_{d,k})$; 
  \item [(c)] what values $i$ sends to an agent $j \in T$
  after computing $g_{k,\ell}(h_{d,k})$ (if 
    $f_{k,\ell}(h_{d,k}) = j$). 
\end{itemize}

Clearly, the number of times that an honest agent invokes VSS
affects the adversary's simulated history. For
(b) and (c), if the adversary assumes that the honest agents have
arbitrary inputs, the values received by agents in $T$ will also be
arbitrary, as opposed to being correlated to the agents' inputs (e.g.,
the messages sent by the mediator can depend in the messages
received by honest agents, which ultimately depend on their inputs). 

\commentout{
Thus, the decisions and actions taken by honest agents
are the only parts of $\vec{\pi}'$ that the adversary cannot simulate,
since the adversary does not know the honest agents' inputs. However,
if 
}
We show next that (a), (b) and (c) are the only decisions and actions
taken by honest agents that the adversary cannot simulate. 
Suppose that 
the adversary had an oracle that could tell
the adversary the number of times each honest agent $i$ invokes VSS
each time $i$ is scheduled, and the values of
each instance of $f_{k,\ell}(h_{d,k})$ and $g_{k,\ell}(h_{d,k})$ (if
$f_{k,\ell}(h_{d,k}) \in T$). Then the adversary could perform
its simulation even without the honest agents' inputs:
it could run its simulation with arbitrary inputs for honest
agents. Whenever an honest agent $i$ is scheduled after receiving a
$\proceed_i$ message, it could ask the oracle how many times $i$
invokes VSS, and could simulate $i$ performing that many
invocations of
VSS
with 
arbitrary inputs, 
 even without knowing the actual local history of $i$.
Similarly, whenever honest agents have to
broadcast or send an agent in $T$ their share of
$f_{k,\ell}(h_{d,k})$ or $g_{k,\ell}(h_{d,k})$, the adversary could ask the
oracle what value $f_{k,\ell}(h_{d,k})$ (resp., $g_{k,\ell}(h_{d,k})$) takes
in the actual history of $\vec{\pi}'$. In its simulation, the adversary
takes the set of shares that the honest agents broadcast or sent to
agents in $T$ to be $f_{k,\ell}(h_{d,k})$-realizable or
$g_{k,\ell}(h_{d,k})$-realizable, respectively, regardless of the local
history of the honest agents in the simulation (we will
show when we present the more detailed construction how this can be done in
such a way that the adversary's local history in its simulation is
still consistent, despite the fact that in the simulation, honest
agents may send different
shares than the ones they computed). 
It follows from Lemma~\ref{lemma:secret} that the 
adversary's local histories in the simulation with this oracle and its
histories in a real
interaction where honest players play $\vec{\pi}'$ with input
$\vec{x}$ are identically distributed. 

Unfortunately, the adversary does not have access to such an
oracle. However, by the construction of $\vec{\pi}'$, the values given by
the oracle can be deduced from the history of the protocol with the
mediator  that the adversary is 
simulating, even without the benefit of an oracle.
Specifically, 
if, for all honest agent $i$, each of the simulated
histories $h_{i,k}$ in $\vec{\pi}'$ is equal to $i$'s
local history at the end of $i$'s $k$th turn in $\vec{\pi}$, and
if each of the histories $h_{d,k}$ is equal to the local history of
the mediator at the end of its $k$th turn;  
the number of times that an honest agent $i$ invokes VSS
after receiving its $k$th $\proceed_i$ message but before receiving
its $(k+1)$st  
$\proceed_i$ message in $\vec{\pi}'$ is the number of messages sent by $i$
during its $k$th turn in $\vec{\pi}$; $f_{k,\ell}(h_{d,k})$ is the
recipient of the $\ell$th message sent by the mediator during its
$k$th turn; and $g_{k,\ell}(h_{d,k})$ is the content of this message
(which is known by the adversary if its recipient is in $T$).
Thus, if the adversary could schedule agents and deliver messages in
$\vec{\pi}$ in such a way that, for each agent $i$ and  all $k$, the local
histories $h_{i,k}$ in the adversary's simulation are the same as 
the local history of $i$ in $\vec{\pi}$ at the end of $i$'s $k$th
turn, then we could dispense with the oracle.  However, because the
adversary does not know in the honest agents' input profiles, it
cannot in general do this.

Fortunately, we do not need quite this much. 
\commentout{
By 
Lemma~\ref{lemma:secret}, the actual values that honest agents share
using VSS and compute using CC have no effect on  for the
adversary's local history 
(the local history profile of malicious players and of the adversary)
in its simulation.
}
Recall that the aim of the simulation is for the adversary to compute
what history it would have in $\vec{\pi}'$. Since 
$|T| < n/4$, Lemma~\ref{lemma:secret}
implies that the local histories of the agents in $T$ and of the
scheduler have the same
distribution, independent of which values are being shared by honest agents.
Therefore, 
to deduce the values given by the oracle, 
 it suffices 
 that the adversary schedules agents and
 delivers messages 
 in $\vec{\pi}$ in such a way that, for each agent 
 $i \not \in T$,
  the local
histories $h_{i,k}$ in the adversary's simulation are the same as 
the local history of $i$ in $\vec{\pi}$ at the end of $i$'s $k$th
turn, except possibly for $i$'s input and the content of the messages
sent and received by $i$.
This means that the local histories $h_{i,k}$ and $i$'s local
histories in $\vec{\pi}$ at the end of its $k$th turn should consist
of exactly the same $\Sch$, $\Snd$, $\Rec$, $\Comp$, and $\Done$ events, and may
differ only in the content of their $\Rec$ and $\Snd$ events. 
 A more
 detailed construction of the adversary $(T,\vec{\tau}_T, \sigma_e)$ is
 given next.  

For simplicity, we assume the adversary is a single entity that
controls both the scheduler and the subset $T$ of malicious agents. Given
$\vec{\pi}, \vec{\tau}'_T, \sigma'_e$, and $\vec{x}_T$, the adversary
$(T,\vec{\tau}_T,\sigma_e)$.
starts by performing a simulation of a history of $(\vec{\pi}'_{-T},
\vec{\tau}'_T, \sigma'_e, \vec{0}/_{(T,\vec{x}_T)})$ using additional
variables $\alpha_{i,k,\ell}$ and $\beta_{i,k,\ell}$ with $i,k,\ell
\in \mathcal{N}$, all initially set to a special value
$\bot$. Whenever one of the following events occur in the
simulation, the adversary proceeds as described below.
%
\begin{itemize}
  \item[(1)] An agent $i \not\in T$ is scheduled 
 after receiving
  a $\proceed_i$ message: In this case, $\sigma_e$ schedules
   $i$ in $\vec{\pi}$. If $i$ sends $\ell$ messages when it is scheduled in
  $\vec{\pi}$, then in 
  its simulation,
   $i$ is scheduled and
    invokes VSS $\ell$ times with input $0$, 
    regardless of $i$'s local history $i$ in the simulation.
  \item[(2)] An honest agent $i \not \in T$ terminates the share phase of 
    $(VSS, j, \ell)$, with $j \in T$: 
Since $n > 4t$,  
the properties of VSS guaranteee that at least
  $2t+1$ honest agents in the simulation will also compute
  their share of $(VSS, j, \ell)$, and that these shares  
    reconstruct a unique value $mes_{j,\ell}$. The scheduler $\sigma_e$ then
schedules agent $j$ in $\vec{\pi}$, and $j$ sends message
$mes_{j,\ell}$ to the mediator, tagged with label $\ell$. 
\item[(3)] 
  An agent $i \not \in T$ is the first honest agent to compute 
  the permutation $\theta_k$ computed in phase $k2$.
\commentout{  
  (in phase $k2$ when computing
    $h_{d,k}$).
    } 
Let  $m_k$ be the 
cardinality of the domain of  
 of $\theta_k$, and let $(VSS,
  j_1, \ell_1), \ldots, (VSS, j_{m_k}, \ell_{m_k})$ be the
  invocations of VSS that are included in $h_{d,k}$ in Phase
    $k1$. The scheduler $\sigma_e$ delivers agent $j_{\theta_k(1)}$'s
  $\ell_{\theta_k(1)}$th message, ..., and $j_{\tau(n)}$'s $\ell_{\theta_k(n)}$th
  message to the mediator in $\vec{\pi}$, and then schedules the
  mediator.
  \item[(4)] An agent $i \not \in T$ terminates the 
invocation of CC instance for $f_{k, \ell}(h_{d,k})$.
    If the additional variable $\alpha_{i,k,\ell} \not = \bot$,
then the adversary continues 
its simulation
under the
assumption that 
$i$'s output of the CC invocation is $\alpha_{i,k,\ell}$ rather than the
actual output. 
 Otherwise,
since $n > 4t$, 
the assumption made at the end of
Section~\ref{ssection:circuit-computation}  
 guarantee
 that there is a set 
 $I$ 
 of at least $2t+1$ honest agents in the
 simulation that 
 have already computed 
  their shares 
   of $f_{k,\ell}(h_{d,k})$.
\commentout{
the shares of these agents in $S$ uniquely determine the output
$s_{j,k,\ell}$ that each other agent $j$ would get if they terminate
the CC procedure.
}
Let $\{s_{i,k,\ell}\}_{i \in [n]}$ be the unique full extension of the
shares of the agents in $I$, and let $j_{k,\ell}$ be either the receiver
of the mediator's $\ell$th message during its $k$th turn or 0 if the
mediator didn't send $\ell$ messages during its $k$th turn. The
adversary 
samples
uniformly at random a full $j_{k,\ell}$-extension
$\{s'_{i,k, \ell}\}_{i \in [n]}$ of $\{s_{i,k,\ell}\}_{i \in T}$ and
sets $\alpha_{i,k,\ell}$ to $s'_{i,k,\ell}$ for all $i \in [n]$.
\commentout{
The adversary sets $\alpha_{i,k,\ell}$ to
$s_{i,k,\ell}$ for each $i \in T$, then  chooses values 
for $\{\alpha_{i,k, \ell}\}_{i \not \in T}$ uniformly at random that
satisfy the constraint that
the output of 
$\longvec{\VSS}^{rec}(\alpha_{1,k,\ell}, \ldots, \alpha_{n, k, \ell})$
is
either the receiver of $\ell$th message sent by the mediator the 
$k$th time it is scheduled  
in $\vec{\pi}$, or 0 if the mediator didn't send $\ell$
messages 
the $k$th time it was scheduled. 
}
The adversary continues 
its simulation
 by assuming that 
 $i$'s output of the CC invocation is $\alpha_{i,k,\ell}$ rather than
 $s_{i,k,\ell}$. 
 \item[(5)] An agent $i \not \in T$ terminates the CC invocation for
  $g_{k, \ell}(h_{d,k})$, and $k$ and $\ell$ are such that
\commentout{
\begin{itemize}
\item [(a)] The mediator sent at least $\ell$ messages
 the $k$th time it was scheduled in $\vec{\pi}$. 
\item [(b)] The recipient $j_{k,\ell}$ of the mediator's $\ell$th
    message in $\vec{\pi}$ satisfies $j_{k,\ell}
  \in T$. 
\end{itemize}
}
the mediator sent at least $\ell$ messages the $k$th time it was
scheduled in $\vec{\pi}$, and the recipient $j_{k, \ell}$ of the
$\ell$th message sent by the mediator the $k$th time it is scheduled in
$\vec{\pi}$ is in $T$. 
If $\beta_{i,k,\ell} \not = \bot$, then the adversary continues 
its simulation
by assuming that 
$i$'s output of the CC invocation is $\beta_{i,k,\ell}$ rather than the actual output.
Otherwise, let $mes_{k,\ell}$ be the
content of the 
$\ell$th message that the mediator sends the $k$th time it is
scheduled (note that this value is known by 
$j_{k,\ell}$,
 and hence by the
adversary). 
 Since $n > 4t$,
  by the properties of CC, 
there is a set 
$I$ 
with at least $2t+1$ honest agents that 
have computed
\commentout{
their shares  of $g_{k, \ell}(h_{d,k})$ in the simulation, 
and the shares of the honest agents in $S$ uniquely determine the output
$s'_{j,k,\ell}$ that each 
agent $j$ 
that terminates the CC procedure gets.
The adversary sets $\beta_{i,k,\ell}$ to
$s'_{i,k,\ell}$ for each $i \in T$, and chooses values for
$\{\beta_{i,k, \ell}\}_{i \not \in T}$ uniformly at random 
so that the output of  
$\longvec{\VSS}^{rec}(\beta_{1,k,\ell}, \ldots, \beta_{n, k, \ell})$
is
$mes_{k,\ell}$.
}
their shares of $g_{k, \ell}(h_{d,k})$ in the simulation. Let
$\{s''_{i,k,\ell}\}_{i \in [n]}$ be the unique full extension of the
shares of agents in $I$. The adversary 
samples
uniformly at random a
full $mes_{k,\ell}$-extension of $\{s''_{i,k,\ell}\}_{i \in T}$ and
sets $\beta_{i,k,\ell}$ to $s''_{i,k,\ell}$ for each $i \in
[n]$. Then the adversary continues 
its simulation
by assuming that 
\commentout{
$i$ sends
$\beta_{i, k, \ell}$ to $j_{k, \ell}$ rather than sending
$s'_{i,k,\ell}$ to 
possibly some other agent.
}
$i$'s output of the CC invocation is $\beta_{i, k, \ell}$ rather than
$s''_{i,k,\ell}$.
\item[(6)] An agent $i \not \in T$ reconstructs
$g_{k,\ell}(h_{d,k})$ using VSS for some $k$ and $\ell$. In
        this case, the scheduler $\sigma_e$ delivers the
$\ell$th message sent by the mediator to $i$ the $k$th time the mediator was
    scheduled.
\end{itemize}


Note this construction for the adversary is well defined. The first
and second clause guarantee that if an agent terminated  $(VSS, i,
\ell)$ in the adversary's simulation, then agent $i$ 
sent a message tagged with label $\ell$ in the corresponding history
of $\vec{\pi}$. Also, whenever 
an honest agent finishes the computation of $f_{k,\ell}(h_{d,k})$ or
$g_{k,\ell}(h_{d,k})$, it must have terminated the computation of
$\theta_k$, and 
(3)
 guarantees that the mediator has been scheduled
at least $k$ times (as needed for 
 (4) and (5)). 

The first step in proving that this construction of $\vec{\tau}_T$ and
$\sigma_e$ satisfies clause (b) of the definition of $t$-bisimulation
is to show that the adversary can simulate how many
times each honest agent invokes VSS in
$\vec{\pi}'$. Recall that honest agents may initiate a new invocation
of only after receiving a $\proceed$ message. 
Given a history $h$ in $\vec{\pi}'$, let $a_{i,k}(h)$ denote the
number of times that agent $i$ invokes VSS after receiving 
its $k$th $\proceed_{i}$ message but before receiving the $(k+1)$st
$\proceed_i$ message. 
Let $\left((a_{i,k}(h))_{i \not \in T}\right)_{k \in \mathbb{N}}$ be
the sequence of all such values,
arranged lexicographically first by their $k$ index and then by their
$i$ index.  
 The following lemma, which follows immediately
from the construction of $\vec{\tau}_T$ and $\sigma_e$, shows that
the distribution of these random variables is the same in $\vec{\pi}'$ and in
the adversary's simulation.
\begin{lemma}\label{lemma:same-VSS}
Fix $T \subseteq [n]$ with $|T| < n/4$ and an input profile
$\vec{x}$. Let $H$ be the distribution over histories when agents use
$(\vec{\pi}'_{-T}, \vec{\tau}'_T)$ with scheduler $\sigma_e'$ and
input $\vec{x}$, and let $H'$ be the distribution over histories in
the adversary's simulation when agents use $(\vec{\pi}_{-T},
\vec{\tau}_T)$ with scheduler $\sigma_e$ and input $\vec{x}$. Then
$S(H)_T$ and $S(H')_T$ are identically distributed. 
\end{lemma}

As we pointed out before, 
one of 
the only decisions made by honest agents that the
adversary cannot simulate without additional information is the
number of invocations of VSS that they perform when they are
scheduled. Thus, this lemma shows that honest agents behave exactly
the same in $\vec{\pi}'$ and in the adversary's simulation except for
the values that they share using VSS
and the values of $f_{k,\ell}(h_{d,k})$ and $g_{k, \ell}(h_{d,k})$
sent to other players. However, by 
Lemma~\ref{lemma:secret}, 
exactly which values are shared using VSS
does not affect the
adversary's local history in its simulation. Therefore, the only
events that might differ between the adversary's local history in
$\vec{\pi}'$ and its local history in its simulation in $\vec{\pi}$
are those in which agents in $T$ receive shares of
$f_{k,\ell}(h_{d,k})$ and $g_{k, \ell}(h_{d,k})$ from honest
agents. In $\vec{\pi}'$, $f_{k,\ell}(h_{d,k})$ and $g_{k,
  \ell}(h_{d,k})$ are the recipient and the content of the mediator's
$\ell$th message during its $k$th turn, which are computed using CC
with the simulated mediator's local history $h_{d,k}$ as input. Since
the adversary assumes in its simulation that the values that honest agents
share using VSS are all 0, 
\commentout{
the local history $h_{d,k}$ that
honest agents compute in the simulation are not identically
distributed to the local history they would compute if they actually
use $\vec{\pi}'$.
}
the local histories $h_{d,k}$ that honest agents compute in the adversary's simulation would not follow the same distribution as their local histories if they used $\vec{\pi}'$ with their actual input. 
 Thus, the shares of
$f_{k,\ell}(h_{d,k})$ and $g_{k, \ell}(h_{d,k})$ would also 
have
 a
different distribution.

However, much as when dealing with VSS
invocations, the adversary can use the mediator's actions in $\vec{\pi}$
as feedback for its simulation, and simulates that the shares of
$f_{k,\ell}(h_{d,k})$ and $g_{k, \ell}(h_{d,k})$ that honest agents
send 
to players in $T$ 
define secrets $j_{k, \ell}$ and $\mu_{k, \ell}$ respectively,
regardless of their local history, where $j_{k,\ell}$ and
$\mu_{k,\ell}$ are the recipient and the content of the mediator's
$\ell$th message during its $k$th turn in $\vec{\pi}$ (note that the
adversary  does this for  
 $g_{k, \ell}(h_{d,k})$ whenever
$j_{k,\ell} \in T$, since otherwise it does not know the content of
this message).
\commentout{
By construction, these shares are still
$j_{k,\ell}$-realizable and $\mu_{k,\ell}$-realizable even if they are
put together with the output of the CC procedure for
$f_{k,\ell}(h_{d,k})$ and $g_{k, \ell}(h_{d,k})$ of agents in $T$
(note that since these agents are malicious, they might not actually
compute such values, however, they can be deduced from
their local histories); otherwise, the local histories of agents in
$T$ would be inconsistent. The following lemma shows the correctness
of this construction.
}
By construction, this simulation proceeds in such a way that the shares
that honest players send are ``consistent'' with the shares of
$f_{k,\ell}(h_{d,k})$ and $g_{k, \ell}(h_{d,k})$ that players in $T$
could compute from their local history, in the sense that the shares
of $f_{k,\ell}(h_{d,k})$ and $g_{k, \ell}(h_{d,k})$ sent by honest
players together with the shares of those functions that players in
$T$ could compute are $j_{k,\ell}$-realizable and
$\mu_{k,\ell}$-realizable respectively. 
Note that if this weren't the case, the simulated local history of
players in $T$ could not occur if honest players played $\vec{\pi}'$. 
 The following lemma shows the correctness
 of this construction.

Given a history $h$ of $\vec{\pi}'$, let $R_{k, \ell}^T(h)$ denote the
subsequence of $\Rec$ events in $h_T$ of messages from agents not in
$T$ involving the computation of $f_{k,\ell}(h_{d,k})$ and $g_{k,
  \ell}(h_{d,k})$, using CC, the broadcast procedures for the shares of
  $f_{k,\ell}(h_{d,k})$, and the messages in which they send their
shares of $g_{k, \ell}(h_{d,k})$.  
\commentout{
Then, the following lemma holds, which together with Lemma~\ref{lemma:secret} and  Lemma~\ref{lemma:same-VSS} implies the desired result.
}
\begin{lemma}\label{lemma:same-CC}
    Fix $T \subseteq [n]$ with $n > 4|T|$ and an input profile
  Let $\vec{\pi}_h(\vec{x}, A)$ be the distribution over histories
  when agents use $\vec{\pi}$ with adversary $A$ and input $\vec{x}$,
  and let $A := (T, \vec{\tau}_T, \sigma_e))$ and $A' := (T,
  \vec{\tau}'_T, \sigma'_e)$. Then $R_{k,
    \ell}^T(\vec{\pi}'_h(\vec{x}, A')$ and $R_{k,
    \ell}^T(\vec{\pi}_h(\vec{x}, A))$ are identically distributed 
for all $k, \ell$,
where $R_{k, \ell}^T(\vec{\pi}_h(\vec{x}, A))$ is the distribution on
sequences of 
$\Rec$ events defined naturally by composing $\vec{\pi}_h(\vec{x}, A)$
and $R_{k, 
  \ell}^T$. 
\end{lemma}

\begin{proof}
The proof of this lemma is analogous to that given by Canetti in his 
  proof of Lemma 4.31 \cite[p. 91]{canetti96studies}.
\end{proof}

As we noted in Section~\ref{sec:uniform-CC}, we implement CC in such
a way that
there is no correlation between the shares of different circuit
computations. Thus,
 Lemma~\ref{lemma:same-CC}
can be easily generalized to 
show that $(R_{k,\ell}(H))_{k,\ell \in \mathbb{N}}$ and
$(R_{k,\ell}(H'))_{k,\ell \in \mathbb{N}}$ are identically
distributed. This, together with Lemma~\ref{lemma:secret} and
Lemma~\ref{lemma:same-VSS} implies Theorem~\ref{thm:main}(a) in the
  case that $t=t'$.

We now prove that $\vec{\pi}'$ $(t,t')$-bisimulates $\vec{\pi}$ in the
general setting, where $3t+t' < n$ and $t \ge t'$. 
Let $t'' = t - t'$.
Given a protocol $\vec{\pi}$ for $n$ agents,
consider the protocol $(\vec{\pi},\eta)$ for  
$n+t''$ 
agents, where the
first $n$ agents use $\vec{\pi}$, while the last 
$t''$
agents use the null protocol, that is, they never send any
messages. Given an adversary $A'=  
(T,\vec{\tau}'_T, \sigma_e')$ in the setting with $n$ agents (and no
mediator), consider  
an adversary $A'' = (T, \vec{\tau}'_T, \sigma''_e)$ in the setting with 
$n+t''$
agents in which
$\sigma''_e$ is a relaxed scheduler that acts just like $\sigma'_e$,
except that it might schedule agents in 
$\{n+1, \ldots, n+t''\}$,
although it never delivers their messages 
(note that since $\sigma''_e$ does not deliver the messages sent by
the last 
$t''$
 agents, $\tau'$ is well defined even in the setting with
$n+t''$
 agents). 
Since $4t < n + t''$, following 
  the same
  construction as in the case that $t=t'$,
  there exists an adversary 
$A = (T, \vec{\tau}_T, \sigma_e)$
   such that  
$$(\vec{\pi} + \eta)(\vec{x}, A'') = (\vec{\pi} + \eta)(\vec{x}, A)$$
for all input profiles $\vec{x}$. 
\commentout{
However, $\sigma_e$ is a relaxed
scheduler, since some of the VSS-share and CC procedures in
$\vec{\pi}'_{VSS}$ are not guaranteed to terminate if $\sigma''_e$ is
relaxed.}
Note that the scheduler $\sigma_e$ resulting from this construction is
in fact a relaxed scheduler: since some of the VSS-share and CC
instances in $(\vec{\pi} + \eta)'$ are not guaranteed to terminate if
$\sigma''_e$ is relaxed, some messages might not be delivered by
$\sigma_e$. 

Consider a scheduler $\sigma_e'''$ in
$\vec{\pi}$
 that acts like
$\sigma_e$ except that it does not schedule agents $n+1, \ldots, n+t$.
By construction, 
\commentout{
$$O(\vec{\pi}_{-T}, \vec{\tau}_T,
\sigma_e''', \vec{x}) = O(\vec{\pi}_{-T}, \eta, \vec{\tau}_T, \sigma_e,
\vec{x})$$ 
and $$O(\vec{\pi}'_{-T}, \vec{\tau}'_T, \sigma'_e, \vec{x}) = O((\vec{\pi} + \eta)'_{-T}, \vec{\tau}''_T, \sigma''_e, \vec{x}),$$
so $$O(\vec{\pi}_{-T}, \vec{\tau}_T, \sigma_e''',
\vec{x}) = O(\vec{\pi}'_{-T}, \vec{\tau}'_T, \sigma'_e, \vec{x}),$$
}
$$\vec{\pi}(\vec{x}, (T, \vec{\tau}_T, \sigma_e''')) = (\vec{\pi} + \eta)(\vec{x}, A)$$
and
$$\vec{\pi}'(\vec{x}, A') = (\vec{\pi} + \eta)(\vec{x}, A'')$$
so
$$\vec{\pi}(\vec{x}, (T, \vec{\tau}_T, \sigma_e''')) = \vec{\pi}'(\vec{x}, A')$$
 as
desired. 

Finally, it remains to show that if $|T| \le t'$, $\sigma'''_e$ is not relaxed. 
Given the adversaries $A'$ and $A''$ defined above, consider an
adversary $A^* := (\vec{\tau}^*_{T'}, \sigma^*_e)$ for $(\vec{\pi} +
\eta)'$, such that $T' = T \cup \{n+1, \ldots, n+t''\}$, agents $i \in
T$ use $\vec{\tau}'_i$, agents in $\{n+1, \ldots, n+t''\}$ send no
messages, and $\sigma_e^*$ acts just like $\sigma''_e$. By construction, 
$$(\vec{\pi} + \eta)'(\vec{x}, A^*) = (\vec{\pi} + \eta)(\vec{x}, A'')$$

In this case, $\sigma^*_e$ is not a relaxed scheduler, since agents in
$\{n+1, \ldots, n+t''\}$ never send any messages. Moreover, since
$|T'| = |T| + t'' < t$, it follows that $4|T'| < n$, and
reasoning analogous to the previous case shows that there exists an
adversary $A
= (T,\vec{\tau}, \sigma'''_e)$ such that 
$$\vec{\pi}(\vec{x}, A) = \vec{\pi}'(\vec{x}, A').$$ 
%
However, in this case, $\sigma'''_e$ is not relaxed, since $\sigma''_e$ was not.

\subsection{Bounding the number of messages}\label{sec:message-bound}

As mentioned in Section~\ref{sec:construction}, our construction of
$\vec{\pi}'$ does not bound the number of messages sent. To see
this, note that players compute $h_{d,k}$ each time that the mediator 
is scheduled in the simulation. Since the number of times that the mediator
can be scheduled is unbounded, the number of messages sent in
$\vec{\pi}'$ can be unbounded as well.

If the mediator $\pi_d$ is responsive, we show how we can modify the construction of
Section~\ref{sec:construction} so as to bound the number of
messages. The idea is that, since $\pi_d$ is responsive, agents
don't need to simulate all the mediator's histories; it suffices to
simulate only the histories in which that the mediator receives at
least one message at every turn
except possibly the first one.
Note that this bounds the number of mediator turns
that the players simulate by $N$, and thus guarantees that the
expected number of messages in $\vec{\pi}'$ is polynomial in $n$ and
$N$ since all primitives satisfy this property. To do this, agents run
Section~\ref{sec:proof} with a simple modification in the
computation of $h_{d,k}$
for $k > 1$. 
%
Instead of running a consensus protocol $p_{j, \ell, k}$ for each VSS
invocation (VSS, $j$, $\ell$), players use an ACS computation
$C_k$ with parameter $m = 1$, in which the accumulative set $U_i$ of
player $i$ consists of the pairs $(j, \ell)$ such that $i$ has
terminated (VSS, $j$, $\ell$) but $(j, \ell)$ was not
in any core set $C_{k'}$ with $k' < k$. Players then continue
Phases $k2$ to $k4$ as usual, but take $p_{j, \ell, k} =
1$ iff $(j, \ell) \in C_k$.  Since the players take the parameter $m$
to be 1, $|C_k| \ge 1$; thus, it is guaranteed that in the simulation,
the mediator has received at least one message
at its $k$th turn.

The proof of correctness of this modified construction is identical to
that given in Section~\ref{sec:proof} for the original construction.

\subsection{The proof of Theorem~\ref{thm:main}(b)}\label{sec:cotermination-proof}

\commentout{
Sketch of the proof: We prove that each of the sub-protocols used satisfies this property.
}
If $\vec{\pi} + \pi_d$ is in canonical form, the construction of
$\vec{\pi}'$ is as in Section~\ref{sec:construction}, except that if
an honest player reconstructs a message containing ``STOP'', it
terminates. 

Suppose that a set $I$ of at least $2t+1$ honest agents
terminate. This means that all agents in $I$ have computed their
share of each of the mediator's messages.
Thus, for each 
message $\mu$ sent by the mediator, each honest agent $i$ will eventually
receive a subset 
$S_{\mu}^I$ of shares such that $(I, S_{\mu}^I)$ is
$\mu$-realizable. Recall that we assumed (at the end of
Section~\ref{sec:tools}) that the secret-sharing scheme used in $\pi'$
is $t$-determinate.  Thus,
this subset of shares suffices for each honest agent to
uniquely reconstruct $\mu$, even with an adversary of size $t$: if a
pair $(I,S)$ with $|I| \le 2t+1$ is $\mu$-realizable, at least $t+1$
agents from $I$ are honest, and their shares uniquely define $\mu$
(and each of the other agents' shares). Thus, receiving a realizable
set of at least $2t+1$ shares uniquely determines the secret being
shared. 
\subsection{The proof of Theorem~\ref{thm:main-eps}}
The protocol $\vec{\pi}'$ for Theorem~\ref{thm:main-eps} is analogous
to that for 
Theorem~\ref{thm:main}, except that we use the VSS and CC 
implementations of BKR instead of those of BCG.
The proof of Theorem~\ref{thm:main-eps}(a) is then identical to 
that of Theorem~\ref{thm:main}(a).  Since it can be easily shown
that the VSS and CC implementations constructed by BKR
$\epsilon$-$(t,t+1)$-coterminate, Theorem~\ref{thm:main-eps}(b)
follows. 
\commentout{
It can
be easily shown that both of the schemes constructed by BKR
$\epsilon$-$(t,t+1)$-coterminate.
} 


\section{Conclusion}\label{sec:conclusion}

We have shown how to simulate arbitrary protocols securely in an asynchronous
setting in a ``bidirectional'' way (as formalized by our notion of
bisimulation).  This bidirectionality plays a key role our application
of these results in a companion paper; we believe that it might turn
out to be useful in other settings as well.  While this property
holds for the BCG function simulation, proving that we can simulate
arbitrary protocols so that it holds seems to be nontrivial.

\commentout{
Our construction may not be message-efficient.  Indeed, 
the number of messages required to simulate a protocol
$\vec{\pi} + \pi_d$ using our construction may be unbounded. To see 
this, note that players compute $h_{d,k}$ each time that the mediator
is scheduled in the simulation. Since the number of times that the mediator
can be scheduled is unbounded, the number of messages sent in
$\vec{\pi}'$ can be unbounded as well.
}
Our construction may not be message-efficient in the general
case. However, for responsive mediators, a small modification allows
us to bound the expected number of messages by a function that is
polynomial in the 
number of players $n$ and the maximum number of messages $N$ sent in
the setting with the mediator,
and linear in $c$, the number of gates in a circuit that implements
the mediator's protocol.
It is still an open problem whether all protocols $\vec{\pi} + \pi_d$
can be implemented in a way that the  
expected number of messages sent by honest agents is bounded by some
function of $n$, $N$, and $c$.

\commentout{

}
\bibliographystyle{chicagor}
\bibliography{joe,game1}

\end{document}